\documentclass[11pt,a4paper]{article}

\usepackage[top=20mm,bottom=20mm,left=20mm,right=20mm]{geometry}

\usepackage[dvips]{graphicx}
\usepackage[cp1251]{inputenc}
\usepackage[T2A]{fontenc}
\usepackage[english]{babel}

\usepackage{hyperref}

\usepackage{amsmath}
\usepackage{amssymb}
\usepackage{amsxtra}
\usepackage{amsthm}
\usepackage{latexsym}
\usepackage{array}
\usepackage{multirow}
\usepackage{hhline}
\usepackage{color}
\usepackage{longtable}

\numberwithin{equation}{section}

\newcounter{myta}
\numberwithin{myta}{section}
\newcommand{\myt}{\refstepcounter{myta}\themyta}
\newcommand {\fts}[1] {{\small #1}}
\theoremstyle{plain}
\newtheorem{theorem}{Theorem}
\newtheorem{lemma}{Lemma}
\newtheorem{propos}{Proposition}
\newtheorem{remark}{Remark}
\newcommand {\tx}[1] {{\textrm{#1}}}

\def\mbe{\mathbf{m}}
\def\imm{p}
\def\sigh{\mathcal{S}(h)}
\def\csigh{\widetilde{\mathcal{S}}(h)}
\def\cofb{0.7}
\def\cofs{0.62}
\def\Q{{\rm Sep}}

\def\eps{\varepsilon}
\def\vk{\varkappa}
\def\rk{\operatorname{rank}}
\def\bR{\mathbb{R}}
\def\rmi{\mathrm{i}\,}
\def\mstrut{\vphantom{\bigl(}}
\def\mF{\mathcal{F}}
\def\ds{\displaystyle}
\def\mP{\mathcal{P}}
\def\mK{\mathcal{K}}
\def\mM{\mathcal{M}}
\def\mL{\mathcal{L}}
\def\mD{\mathcal{D}}
\def\mm{\mathcal{M}_1}
\def\mn{\mathcal{M}_2}
\def\mo{\mathcal{M}_3}
\def\ml{\mathcal{M}_4}
\def\pov{\Pi}
\def\vp{\varphi}

\begin{document}

\title{Topological atlas of the Kovalevskaya top in a double field}

\author{M.\,P.~Kharlamov\footnote{Russian Academy of National Economy and Public Administration. E-mail: mikeh@inbox.ru}, P.\,E.~Ryabov\footnote{Financial University under the Government of Russian Federation. E-mail: orelryabov@mail.ru} }

\date{Submitted to J. Fundam. and Appl. Math. 13.08.2014}

\maketitle

\begin{abstract}
We fulfill the rough topological analysis of the problem of the motion of the Kovalevskaya top in a double field. This problem is described by a completely integrable system with three degrees of freedom not reducible to a family of systems with two degrees of freedom. The notion of a topological atlas of an irreducible system is introduced. The complete topological analysis of the critical subsystems with two degrees of freedom is given. We calculate the types of all critical points. We present the parametric classification of the equipped iso-energy diagrams of the complete momentum map pointing out all chambers, families of 3-tori, and 4-atoms of their bifurcations. Basing on the ideas of A.T.\,Fomenko, we introduce the notion of the simplified net iso-energy invariant. All such invariants are constructed. Using them, we establish, for all parametrically stable cases, the number of critical periodic solutions of all types and the loop molecules of all rank~1 singularities.

\end{abstract}

\tableofcontents

\section{Introduction}\label{sec1}

The investigation of bifurcations in integrable Hamiltonian systems with two degrees of freedom and a non-linear additional integral was started in \cite{Kh1976,Kh1979,KhPMM,KhDAN,KhBook}. A complete theory of such systems was built by A.T.\,Fomenko, H.\,Zieschang, and A.V.\,Bolsinov (see \cite{BolFom}). A theory of topological invariants for integrable systems with many degrees of freedom was created by A.T.\,Fomenko in the works \cite{fom91,fomams91}. At the same time, due to the absence of practical methods to analyze non-trivial examples, up to this moment no investigation has been fulfilled of a three dimensional topology of a system not reducible to a family of systems with two degrees of freedom.

The main source of examples of integrable systems with two degrees of freedom is the problem of the motion of a rigid body about a fixed point in an axially symmetric force field and its generalizations to some Lie co-algebras. Non-symmetric force fields lead to the systems with three degrees of freedom without a possibility of the global order reduction. This type of problem is obtained as the generalization of the classical S.V.\,Kovalevskaya case to the motion of a gyrostat in a double force field. The complete integrability of this system is proved in the works \cite{Bogo,Yeh,ReySem} by consequent generalizations of the classical integrals. Here, we consider the system supposing that the gyrostatic momentum is zero. This case usually is called the Kovalevskaya top in a double field. We give the complete investigation of the three-dimensional phase topology partially announced in \cite{KhRy2012}. The knowledge of basic definitions and facts on the singularities of momentum maps and bifurcations in the case of many degrees of freedom \cite{Fom1986,Fom1989,BolsOsh2006} is supposed. The corresponding theory, in details, is given in~\cite{BolFom}.

The structure of the present paper is based on the idea of a topological atlas of an irreducible system with three degrees of freedom depending on a set of some physical parameters~\cite{KhGDIS,KhVVMSH13}. In this case, the rough iso-energy invariant is not any more a one-dimensional graph, but can be represented as a so-called equipped iso-energy diagram, which is the bifurcation diagram of the restriction of the momentum map to an energy level; the diagram is stratified by the rank of the momentum map and the types of critical points in the pre-image.  Moreover, the notion of an equipped diagram includes its span, i.e., it is supplemented with two-dimensional chambers which are the components of the plane of the constants of additional integrals cut out by the bifurcation diagram together with the number of the regular tori in the chambers and the way in which these tori are united in the so-called families. To this end, one needs to obtain a complete classification of the iso-energy diagrams and analyze their evolution with respect to the parameters of the system.

First, we obtain the stratification of the phase space by the critical subsystems. These subsystems are formed by the sets of critical points lying in the pre-images of the smooth surfaces bearing the bifurcation diagram $\Sigma$ of the complete momentum map in three-dimensional space. The critical subsystems have less than three degrees of freedom, thus all up-to-date experience of the topological analysis is applicable to them. In particular, we calculate the bifurcation diagrams of the critical subsystems, the types of points of ranks 0 and 1 which are critical \textit{inside} the subsystems, and also the corresponding atoms of two-dimensional bifurcations. At the same time, the equations of the phase spaces of the critical subsystems provide a way to calculate explicitly the so-called \textit{external} type of each point of any critical subsystem, in particular, for all points of rank 2. Combining this information, we obtain the complete classification of the critical points with respect to their types in the initial system with three degrees of freedom.

Let $c$ be the vector of physical parameters, $H$ the Hamiltonian of a system. In the non-resonant case, if there are no exotic bifurcations (having the so-called splitting atoms) the cross sections of the bifurcation diagram $\Sigma$ by the planes of the constant values $H=h$ undergo transformations if either the iso-energy manifold $H_h=\{H=h\}$ has a non-empty intersection with the set of points of rank 0 and degenerate points of rank 1, or the Hamiltonian has an extremum on the set of degenerate points of rank 2. For a fixed value of the vector $c$ these conditions define only a finite set of points. Writing out the dependencies at such points $h=h(c)$, we obtain the set of surfaces in the space of the parameters $(c,h)$ which separate inequivalent parametrically stable iso-energy diagrams. Thus, to obtain the \textit{topological atlas} we need to construct the set of separating surfaces and, for each component of the supplement to the separating set in the space of $(c,h)$, to point out the corresponding equipped diagram. Further, comparing these diagrams we also fulfil the description of the families of regular tori.

The experience of investigating the integrable systems, even in the case of two degrees of freedom, shows that the bifurcation diagrams and the chambers generated by them usually have a row of specific singularities connected with the existence of ``extra small'' regions. One of the solutions of this problem is to pass from the real forms of the diagrams to the net topological invariants of A.T.\,Fomenko~\cite{fom91}. While applying to the concrete systems, it is reasonable to consider this invariant in some simplified form. Let us represent the span of an iso-energy diagram (which is the image of the corresponding iso-energy manifold depending on $h$ and $c$) as a two-dimensional cell complex. For non-separating values $(c,h)$ we supply all 1-cells with the notation of the corresponding 4-atom and all 2-cells (chambers) with the number of the Liouville tori in the pre-image. The result is also one of the possible forms of the rough topological invariant of the system. Passing to the conjugate complex makes this invariant even more clear. In it, the 0-cells are supplied with natural numbers (the number of 3-tori), 1-cells obtain the notation of the 4-atoms, and the 2-cells represent critical closed orbits (critical motions built of the points of rank 1). Adding one 0-cell with the number 0 that represents the chamber with empty integral manifolds, we turn the conjugate complex to a cellular partition of the two-dimensional sphere. We call this equipped complex a simplified net invariant of the restriction of the initial system to an iso-energy level.

Having the full set of net invariants let us analyze the atoms at the boundary of each 2-cell. From this we obtain the complete information on the stability of the corresponding closed orbits and construct the loop molecule (in the rough sense). For non-degenerate orbits of low complexity basing on the corresponding classification \cite{BolFom} we can also establish the exact topology of loop molecules.

The result of such investigation can be called a \textit{complete topological atlas} of an integrable system. In what follows, we construct such an atlas for the Kovalevskaya top in a double field.

\section{Equations and integrals. The notion of a critical subsystem}\label{sec2}
%%%%%%%%%%%%%%%%%%%%%%%%%%%%%%%%%%%%

The problem of the motion of the Kovalevskaya top in a double field is described by the system of equations \cite{Bogo}
%%%%%%%%%%%%%%%%%%
\begin{equation}\label{eq2_1}
\begin{array}{lll}
\displaystyle{2 \, \dot \omega_1=\omega_2\omega_3+\beta_3,}
&\displaystyle{ \dot \alpha_1=\alpha_2\omega_3-\alpha_3\omega_2,} &
\displaystyle{\dot\beta_1=\beta_2\omega_3-\omega_2\beta_3,}\\
\displaystyle{2 \, \dot \omega_2 = - \omega_1 \omega_3-\alpha_3,} &
\displaystyle{\dot \alpha_2=\omega_1\alpha_3-\omega_3\alpha_1,}&
\displaystyle{\dot\beta_2=\omega_1\beta_3-\omega_3\beta_1,}\\
\displaystyle{\dot\omega_3=\alpha_2-\beta_1,}& \displaystyle{\dot
\alpha_3=\alpha_1\omega_2-\alpha_2\omega_1,} & \displaystyle{\dot
\beta_3=\beta_1\omega_2-\beta_2\omega_1.}
\end{array}
\end{equation}
%%%%%%%%%%%%%%%%%
Here $\boldsymbol\omega$ is the angular velocity vector. The constant in space vectors $\boldsymbol \alpha, \boldsymbol\beta$ characterize the action of the forces. We choose the notation to have $|\boldsymbol\alpha|\geqslant|\boldsymbol\beta|.$ As shown in \cite{Kh34}, without loss of generality we may consider the force fields mutually orthogonal. Then the geometric integrals of system \eqref{eq2_1} can be written in the form ($a \geqslant b \geqslant 0$)
\begin{eqnarray}\label{eq2_2}
|{\boldsymbol\alpha}|^2=a^2,\quad |{\boldsymbol\beta}|^2=b^2,\quad
{\boldsymbol\alpha}\cdot {\boldsymbol\beta}=0.
\end{eqnarray}

Using the components of the kinetic momentum $M_1 =2 \omega_1$, $M_2 =2 \omega_2$, $M_3 = \omega_3$, we take to the space ${\bR}^9({\boldsymbol \omega,\boldsymbol\alpha, \boldsymbol\beta})$ the Lie\,--\,Poisson bracket of the Lie co-algebra $e(3,2)^*=\{(\mathbf{M},\boldsymbol{\alpha}, \boldsymbol{\beta})\}$. This bracket was introduced in \cite{Bogo}. Then equations \eqref{eq2_1} take the form $\dot x=\{H,x\}$, where $x$ is any of the coordinates and
\begin{equation}\label{eq2_5}
H=\frac{1}{2}(2\omega_1^2+2\omega_2^2+\omega_3^2)-\alpha_1-\beta_2.
\end{equation}
The Casimir functions of the Lie\,--\,Poisson bracket are the left-handed parts of equations \eqref{eq2_2}. Therefore, the vector field \eqref{eq2_1} restricted to the six dimensional manifold $\mP^6$ defined in ${\bR}^9(\boldsymbol\omega,\boldsymbol\alpha,\boldsymbol\beta)$ by these equations is a Hamiltonian system with three degrees of freedom.

For $b=0$ system \eqref{eq2_1} describes the case of S.V.\,Kovalevskaya of the motion of a rigid body in the gravity field. If $a=b$, then we have the case of H.M.\,Yehia~\cite{Yeh}. These two limit problems admit a symmetry group and are reduced to families of integrable systems with two degrees of freedom having the 2-sphere as the configuration space. The classical Kovalevskaya case is studied in
\cite{KhPMM,KhDAN,BolFomRich}. The Yehia case and its generalizations were considered in \cite{ZotMTT,JPA}. The global approach to the classification of the integrable systems on the 2-sphere generated by the problems of the rigid body dynamics with axially symmetric potentials is realized in the works \cite{Fom1991,BolFom1994,BolKozFom1995}. In the sequel we consider the non-symmetric case $a>b>0$, which is not reducible to two degrees of freedom.

The first integrals of system \eqref{eq2_1} found in \cite{Bogo} and \cite{ReySem}
\begin{equation}\label{eq2_6}
\begin{array}{l}
K=(\omega_1^2-\omega_2^2+\alpha_1-\beta_2)^2+(2\omega_1\omega_2+\alpha_2+\beta_1)^2,\\
G=\left[\omega_1\alpha_1+\omega_2\alpha_2+\frac{1}{2} \alpha_3
\omega_3\right]^2 +
\left[ \omega_1 \beta_1 + \omega_2 \beta_2 + \frac{1}{2} \beta_3 \omega_3 \right]^2+\\
\phantom{G=}+\omega_3\left[(\alpha_2\beta_3-\alpha_3\beta_2)\omega_1+(\alpha_3\beta_1-\alpha_1\beta_3)\omega_2+
\frac{1}{2}(\alpha_1\beta_2-\alpha_2\beta_1)\omega_3\right]-\\
\phantom{K=}-\alpha_1b^2-\beta_2 a^2
\end{array}
\end{equation}
together with $H$ form on $\mP^6$ a complete involutive set. We define the corresponding momentum map ${\mF }: \mP^{6} \to {\bR}^3$ as ${\mF }(x)=\bigl( G(x), K(x), H(x)\bigr)$.

Let $\mK$ denote the set of critical points of the momentum map, i.e., the points at which $\rk d{\mF }(x)<3$. The set of critical values $\Sigma={\mF }(\mK) \subset{\bR}^3$ is called the {\it bifurcation diagram}. The set $\mK$ is stratified by the rank of the momentum map $\mK =\mK^0 \cup \mK^1 \cup \mK^2$. Here $\mK^r= \{x\in \mP^6 | \rk d{\mF}(x)=r \}$. Accordingly, the diagram $\Sigma$ becomes a cell complex $\Sigma =\Sigma^0 \cup \Sigma^1 \cup \Sigma^2$. On the other hand, practically bifurcation diagrams are described in terms of some surfaces in the space of the integral constants. It is often possible to get the equations of these surfaces (implicit or parametric) even without calculating the critical points; these surfaces are the discriminant sets of some polynomials connected, for example, with singularities of algebraic curves associated with Lax representations. We denote such surfaces by $\pov _i$ and write down the representation $\Sigma= \bigcup_i \Sigma_i$, where $\Sigma_i = \Sigma \cap \pov _i$. The value of such representation is that the critical set $\mK$ turns out to be the union of the naturally arising invariant sets $\mM_i\subset \mK \,\cap \mF^{-1}(\pov_i)$. If $\pov_i$ is given by the regular equations of the type
\begin{equation}\label{eq2_11}
\phi_i(g,k,h) = 0,
\end{equation}
then $\mM_i$ is defined as the set of critical points of the integral $\phi_i(G,K,H)$ belonging to its zero level, while the components of the gradient of the function $\phi_i$ calculated at the point of $\mM_i$ after substituting the values of the integrals $G,K,H$ provide the coefficients of a zero linear combination of the differentials $dG,dK,dH$. At any point of a transversal intersection of two surfaces $\pov_i$ and $\pov_j$ we get two independent zero combinations, therefore at the points of the corresponding intersection $\mM_i \cap\mM_j$ the rank of $\mF$ is equal to~1. Obviously, the points of a transversal intersection of three surfaces (the angles of the bifurcation diagram) are generated by the points where $\rk \mF =0$. The sets $\mM_i$ with the induced dynamics on them are called \textit{critical subsystems}.

The critical subsystems and the equations of the surfaces $\pov_i$ in the problem considered were found in the works \cite{Bogo,Kh32,Kh2005}. The detailed description of the stratification of the critical set by the rank of the momentum map is given in \cite{Kh36}. Also in \cite{Kh36}, in terms of explicit inequalities for the energy constant, the regions of the motion existence on the surfaces $\pov_i$ are pointed out. These regions are the sets $\Sigma_i$ forming the bifurcation diagram. The obtained inequalities provide the set in the parameters space separating different types of the cross sections of the diagram $\Sigma$ by the planes of constant energy, i.e., the types of the bifurcation diagrams of the map $G{\times}K$ restricted to iso-energy surfaces $\{H=h\} \subset \mP^6$. It appears that the critical subsystems are integrable and almost everywhere Hamiltonian systems with less than three degrees of freedom. For them, in turn, the induced momentum map is defined. The bifurcation diagram $\Sigma_i^*$ for the map $\mF|_{\mM_i}$ is identified, in the obvious way, with a subset in the union of the skeletons of the set $\Sigma_i$ of dimensions 0 and 1. Here it is natural to introduce the stratification of $\Sigma_i$ geometrically, considering the existing intersections of the type $\Sigma_i\cap \Sigma_j$. Then the 0-skeleton $\Sigma^1_i$ may also contain the tangency points of two surfaces. In the pre-image of such points the rank of $\mF|_{\mM_i}$ does not decrease, so formally these points
do not belong to $\Sigma_i^*$. Of course, the corresponding points of the set $\mK$ will be \textit{degenerate} critical points of $\mF$, but the system ${\mM_i}$ may not notice this fact. The description of the diagrams $\Sigma_i^*$ and the bifurcations {\it inside} the critical subsystems is completed in the works \cite{Zot,KhSav,Kh2006,Kh2009}. The classification of the points of the set $\mK$ with respect to the complete initial system with three degrees of freedom on $\mP^6$ is obtained in the work \cite{RyKh2012}. In the next two sections we give a short exposition of the necessary results of the above cited papers.

\section{Description of critical subsystems and classes of singularities}\label{sec3}
This section contains the results dealing with finding the critical set and classifying the critical points according to their rank.

For a compact description of the critical subsystems we use the change of variables proposed in~\cite{Kh32}:
\begin{equation*}
\begin{array}{l}
\begin{array}{ll}
x_1 = (\alpha_1  - \beta_2) + \rmi (\alpha_2  + \beta_1),&
x_2 = (\alpha_1  - \beta_2) - \rmi (\alpha_2  + \beta_1 ), \\
y_1 = (\alpha_1  + \beta_2) + \rmi (\alpha_2  - \beta_1), & y_2 =
(\alpha_1  + \beta_2) -
\rmi (\alpha_2  - \beta_1), \\
 z_1 = \alpha_3  + \rmi \beta_3, &
z_2 = \alpha_3  - \rmi \beta_3,
\end{array}\\
\begin{array}{lll}
w_1 = \omega_1  + \rmi \omega_2 , & w_2 = \omega_1  - \rmi \omega_2, &
w_3 = \omega_3.
\end{array}
\end{array}
\end{equation*}

Let us introduce the following functions
\begin{equation*}
\begin{array}{l}
  Z_1=w_1^2+x_1, \qquad Z_2=w_2^2+x_2, \\
  F_1  = \sqrt{x_1 x_2} w_3  - \ds{\frac{(x_2 z_1 w_1  + x_1 z_2
        w_2)}{\sqrt{x_1 x_2}}},\qquad \displaystyle{F_2
        =\frac{x_2}{x_1}Z_1-\frac{x_1}{x_2}Z_2}, \\
  R_1 =\displaystyle{\frac{w_2 x_1+w_1 y_2+w_3 z_1}{w_1}-\frac{w_1 x_2+w_2 y_1+w_3 z_2}{w_2},} \\
  R_2 = \displaystyle{(w_2 z_1+w_1 z_2)w_3^2+\Bigl[\frac{w_2 z_1^2}{w_1}+\frac{w_1 z_2^2}{w_2}+w_1 w_2(y_1+y_2)+}\\
  \phantom{R_2 =} \displaystyle{+ x_1 w_2^2+x_2 w_1^2\Bigr]w_3 +\frac{w_2^2 x_1 z_1}{w_1} + \frac{w_1^2 x_2 z_2}{w_2}+}\\
 \phantom{R_2 =} + \displaystyle{ x_1 z_2 w_2+ x_2 z_1 w_1 +(w_1 z_2-w_2 z_1)(y_1-y_2)}.
\end{array}
\end{equation*}

We define the parameters $p>r>0$, putting $p^2  = a^2  + b^2$ and $r^2  = a^2  - b^2$. Below we use them along with $a$ and $b$ when it is convenient for the sake of brevity.

\begin{theorem}[\cite{Kh2005}]\label{theo1}
The critical set of the map $\mF$ consists of the four critical subsystems $\mM_i$ $(i=1,\ldots,4)$ defined in $\mP^6$ by the following systems of equations:
\begin{equation*}
\begin{array}{l}
\mm:\quad Z_1  = 0, \quad Z_2  = 0,\\
\mn:\quad F_1  = 0, \quad F_2  = 0,\\
\mo:\quad R_1  = 0, \quad R_2  = 0,\\
\ml:\quad w_1 =0,\quad w_2=0, \quad z_1=0,\quad z_2=0.
\end{array}
\end{equation*}

$\mF$-images of the sets $\mM_i$, further denoted by $\pov_i$, in the space $\bR^3(h,k,g)$ of the integral constants $H,K,G$ are described by the following systems of equations
\begin{equation}\label{eq3_9}
\begin{array}{ll}
  \pov_1: \left\{ \begin{array}{l} k=0, \\
    g = \ds{\frac{1}{2}}p^2 h -\ds{\frac{1}{4}} f^2;
    \end{array}\right.
 & \pov_2: \left\{ \begin{array}{l} k=r^4 m^2, \\
g = \ds{\frac{1}{2}} (p^2 h - r^4 m); \end{array}\right.
 \\
  \pov_3: \left\{ \begin{array}{l}
\displaystyle{k = 3 s^2 - 4 h s + p^2 + h^2 - \frac{a^2 b^2} {s^2}}, \\
\displaystyle{g = -s^3 + h s^2 + \frac{a^2 b^2}{s}};
\end{array}\right.
 & \pov_4: \left\{ \begin{array}{l} k=(a \mp b)^2,\\
g=\pm a b h.
\end{array}\right.
\end{array}
\end{equation}
Here $f,m,s$ stand for the constants of the partial integrals $F,M,S$ in the subsystems $\mm,\mn,\mo$ respectively:
\begin{equation}\label{fms}
\begin{array}{l}
  F = w_1 w_2 w_3+z_2 w_1+z_1 w_2, \\
  \displaystyle{M =
\frac{1}{2r^2}(\frac{x_2}{x_1}Z_1+\frac{x_1}{x_2}Z_2)},
 \\
  \displaystyle{S=-\frac{1}{4} \big( \frac {y_2 w_1+x_1 w_2+z_1
w_3}{w_1}+\frac{x_2 w_1+y_1 w_2+z_2 w_3}{w_2} \big).}
 \end{array}
\end{equation}
\end{theorem}

Let us give some comments.

Obviously, the systems of equations (invariant relations) describing the sets $\mn$ and $\mo$ have singularities. Here we in fact consider the closure in $\mP^6$ of the corresponding sets of solutions of these systems in their domain of definition.

The subsystem $\mm$ and the integral $F$ are found in the work \cite{Bogo}, the subsystem $\mn$ and the integral $M$ are pointed out in the work \cite{Kh32}. The subsystems $\mo$ and $\ml$ completing the description of the critical set, together with the integral $S$ that is the analog of the Kovalevskaya variable taking a constant value on the critical motions of the 4th Appelrot class, were found in \cite{Kh2005}. The set $\ml$ given by four equations is a smooth two-dimensional manifold. It is diffeomorphic to the union of two cylinders $S^1{\times}\bR$. The induced Hamiltonian system has one degree of freedom.

The sets $\mm$ and $\mn$ are smooth four-dimensional manifolds, though $\mn$ is non-orientable (see \cite{Zot,KhGEOPHY}). Using the explicit parametric equations of the set $\mo$ obtained in \cite{Kh2007} (see also \cite{RyKh2012}) it can be shown that $\mo$ is a smooth four-dimensional manifold everywhere except for the points common with $\ml$. At these points $\mo$ has a transversal self-intersection which is the two-dimensional manifold $\mo \cap\ml$ with boundary.

According to presentation \eqref{eq3_9}, it is convenient to describe the regions of the existence of critical motions $\Sigma_i$ and the bifurcation diagrams $\Sigma_i^*$  of the critical subsystems in terms of the partial momentum maps. For the first three subsystems these are the maps $\mF_i : \mM_i \to \bR^2$ defined as
\begin{equation*}
\begin{array}{l}
  \mF_1 = F^2{\times}H, \quad  \mF_2 = M{\times}H, \quad \mF_3 = S{\times}H.
\end{array}
\end{equation*}
For the subsystem $\ml$ with one degree of freedom it is natural to put $\mF_4=H: \ml \to \bR$.

For the sequel we need to denote different classes of the critical points and the images of these classes under the momentum maps. The rank of the points always is given with respect to the complete momentum map $\mF$.

\begin{remark}\label{rem1}
From now on we agree to use the same notation for the image of some special critical point or of some definite family of critical points no matter what momentum map from the above defined ones is considered. This will not cause any ambiguity. The only exception is the set of points in the self-intersection of the surface $\pov_3$. The points of this set, after unfolding it to the plane $(s,h)$, get two representations with different values of $s$. Two points on the $(s,h)$-plane which give the same point $(h,k,g)$ will be provided by the upper index ``plus'' or ``minus'' for the larger and smaller values of $s$ respectively.
\end{remark}

It is convenient to emphasize first the classes of motions $Q_i$ in the subsystems $\mM_i$ $(i=1,2,3)$ that correspond to the subsets on which the form induced by the symplectic structure degenerates. It is known (see e.g. \cite{FomSimGeom}) that on a manifold defined as a common level of two independent functions such degeneration takes place on a subset of zeros of the Poisson bracket of these functions. At the points of the subsystems $\mM_i$ the following identities hold \cite{Zot,KhSav,Kh2007}
\begin{equation*}
\begin{array}{l}
\ds  \{Z_1,Z_2\} \equiv -2 \rmi F, \quad \{F_1,F_2\} \equiv -2 \rmi r^2 L, \quad \{R_1,R_2\} \equiv \frac{8 \rmi}{S} U.
\end{array}
\end{equation*}
Here $F$ and $S$ are defined in \eqref{fms}, while $L$ and $U$ are the integrals of the subsystems $\mn$ and $\mo$ respectively and have the form
\begin{equation}\label{eqelu}
\displaystyle{L = \frac{1} {{\sqrt {x_1 x_2 } }}[w_1 w_2  + {{x_1
x_2  + z_1 z_2}} M]}, \qquad U =2 S^4-2 H S^3+a^2b^2.
\end{equation}
Thus, $Q_1=\mm\cap\{F=0\}$, $Q_2=\mn\cap\{L=0\}$, and $Q_3=\mo\cap\{U=0\}$.
Almost all points of these sets have rank 2, but all of them including the finite set of periodic solutions consisting of the points of rank 1, as shown in \cite{RyKh2012}, are degenerate critical points of the momentum map $\mF$.

By definition $Q_i \subset \mM_i$. In addition to that ${Q_1 \subset \mn}$ and ${Q_2 \subset \mo}$, though the restriction of the symplectic structure to $\mn$ is non-degenerate at the points of $Q_1$, and the restriction of the symplectic structure to $\mo$ is non-degenerate at the points of $Q_2$. According to Remark~\ref{rem1}, we denote the images of the sets $Q_1,Q_2,Q_3$ under the complete momentum map $\mF$ and the partial maps $\mF_i$ by $\Delta_1,\Delta_2,\Delta_3$ respectively.

We now classify the critical points with respect to their rank and belonging to the critical subsystems.

The system has exactly four points of rank 0 (see \cite{KhZot}). At these points we obviously have $\alpha_1=\pm a$, $\beta_2=\pm b$, all other components of ${\boldsymbol\alpha}$ and ${\boldsymbol\beta}$, the same as the vector ${\boldsymbol\omega}$, equal zero. We denote these points in the increasing order of the value of $H$ by $p_0,p_1,p_2,p_3$. None of these points belong to $\mm$ and all of them belong to $\mn\cap\mo\cap\ml$. The index $i=0,\ldots,3$ is equal to the Morse index of the function $H$ at these points \cite{KhZot}. We denote by $P_i$ the images of $\imm_i$ under the momentum maps.  The coordinates of $P_i$ are easily calculated from \eqref{eq2_5}, \eqref{eq2_6}. In what follows, only the $h$-coordinates equal to $\mp a \mp b$ are important.

All critical points of rank 1 are organized in nine families of periodic motions denote by $\mD_i$ ($i=1,2,3$) and $\mL_j$ ($j=1,\ldots,6$).

The families $\mD_1,\mD_2,\mD_3$ were first described in \cite{Zot} as the sets of critical points of the pair of integrals $H,F$ on $\mm$. Explicit algebraic expressions of the phase variables in terms of one auxiliary variable connected with the time $t$ by elliptic quadrature are given in the work \cite{Kh361}. In particular, it is proved that the union of these families coincides with the intersection of the subsystems $\mm$ and $\mo$. In \cite{Kh361}, the parametric expressions are proposed of the values of the first integrals (general and partial) on these families, where the parameter is the constant of the integral $S$. Denoting the one-dimensional images of the families $\mD_i$ by $\delta_i$ ($i=1,2,3$), we get the equations
\begin{equation}\label{eq3_20}
\begin{array}{l}
\delta_1: \left\{
\begin{array}{l}
\displaystyle{h=2s-\frac{1}{s}\sqrt{(a^2-s^2)(b^2-s^2)}   }\\[2mm]
\displaystyle{f^2={-\frac{2}{s}\sqrt{(a^2-s^2)(b^2-s^2)}}(\sqrt{a^2-s^2}+
\sqrt{b^2-s^2})^2}\\[3mm]
\displaystyle{g=\frac{1}{s}(s^4-s^2\sqrt{(a^2-s^2)(b^2-s^2)}+a^2b^2)}, \quad s \in
[-b,0)
\end{array} \right. ; \\
\delta_2: \left\{
\begin{array}{l}
\displaystyle{h=2s+\frac{1}{s}\sqrt{(a^2-s^2)(b^2-s^2)}   }\\[2mm]
\displaystyle{f^2={\frac{2}{s}\sqrt{(a^2-s^2)(b^2-s^2)}}(\sqrt{a^2-s^2}-
\sqrt{b^2-s^2})^2}\\[3mm]
\displaystyle{g=\frac{1}{s}(s^4+s^2\sqrt{(a^2-s^2)(b^2-s^2)}+a^2b^2)}, \quad s \in
(0,b]
\end{array} \right. ; \\
\delta_3: \left\{
\begin{array}{l}
\displaystyle{h=2s-\frac{1}{s}\sqrt{(s^2-a^2)(s^2-b^2)}   }\\[2mm]
\displaystyle{f^2={\frac{2}{s}\sqrt{(s^2-a^2)(s^2-b^2)}}(\sqrt{s^2-b^2}-
\sqrt{s^2-a^2})^2}\\[3mm]
\displaystyle{g=\frac{1}{s}(s^4-s^2\sqrt{(s^2-a^2)(s^2-b^2)}+a^2b^2)}, \quad s \in
[a,+\infty)
\end{array} \right. .
\end{array}
\end{equation}
Simultaneously we have in mind that $\mD_i \subset \mM_1$, so here $k=0$.

Denote the points \eqref{eq3_20} having the boundary values of $s$ (respectively, $s=-b, b, a$) by $e_i$ ($i=1,2,3$). On the plane $\bR^2(f^2,h)$ the curves $\delta_i$ have the boundary points $e_i$ on the axis $f^2=0$, the curve $\delta_3$ has a cusp point denoted by $e_4$ and corresponding to the value $s_0$ that is a unique root of the equation
\begin{equation}\label{eq3_21}
3s^8 - 4p^2s^6+6 a^2 b^2 s^4-a^4 b^4 = 0
\end{equation}
on the half-line $s>a$. On the plane $\bR^2(s,h)$ all three curves $\delta_i$ have no singular points (except for the boundary points), and the value $s_0$ corresponds to the minimum of $h$ on $\delta_3$.

The families $\mL_j$ are pendulum type motions (oscillations or rotations) about the principal inertia axes of the body and in the phase space have the form
\begin{equation*}
\begin{array}{l}
{\mL}_{1,2}=\{{\boldsymbol \alpha } \equiv \pm a{\mbe}_1, \;
{\boldsymbol \beta } = b({\mbe}_2 \cos \theta - {\mbe}_3 \sin
\theta ), \;
{\boldsymbol \omega } = \theta ^ {\boldsymbol \cdot}
{\mbe}_1 , \; 2\theta ^{ {\boldsymbol \cdot}  {\boldsymbol \cdot} }
=  - b\sin \theta\},\\
{\mL}_{3,4}=\{{\boldsymbol \alpha } = a({\mbe}_1 \cos \theta + {\mbe}_3 \sin \theta ), \; {\boldsymbol \beta } \equiv  \pm b{\mbe}_2 ,\;
{\boldsymbol \omega } = \theta ^ {\boldsymbol \cdot}  {\mbe}_2 ,
\;
2\theta ^{ {\boldsymbol \cdot}  {\boldsymbol \cdot} }  = -
a\sin \theta\}, \\
{\mL}_{5,6}=\{{\boldsymbol{\alpha }} = a({\mbe}_1 \cos \theta -
{\mbe}_2 \sin \theta ),\;
{\boldsymbol{\beta }} =  \pm b({\mbe}_1
\sin \theta  + {\mbe}_2 \cos \theta ), \; {\boldsymbol{\omega }} =
\theta ^ {\boldsymbol \cdot}  {\mbe}_3 ,\; \\
\qquad\,\,\quad\theta ^{ {\boldsymbol \cdot} {\boldsymbol \cdot} }  =  - (a \pm b)\sin \theta\}.
\end{array}
\end{equation*}
Here ${\mbe}_1 {\mbe}_2 {\mbe}_3$ is the canonical basis in $\bR^3$. The upper sign corresponds to the family with the first number. Bifurcations in the families happen at the singular points $p_i$ ($i=0,\ldots,3$) with the above noted values of $h$, namely, $h=\mp a\mp b$. Let us consider that the manifolds $\mL_i$ include also such special trajectories, i.e., they are in fact the \textit{closure} of the corresponding families of periodic trajectories. It is easily seen that the bifurcations inside the families (the birth of oscillations and the transformation of oscillations into rotations) take place according to the following inclusions: $\imm_0 \in \mL_1\cap\mL_3\cap\mL_5$, $\imm_1 \in \mL_1\cap\mL_4\cap\mL_6$,  $\imm_2 \in \mL_2\cap\mL_3\cap\mL_6$, and $\imm_3 \in \mL_2\cap\mL_4\cap\mL_5$. In particular, the minimal values of $H$ on $\mL_i$ are $-a-b$ on $\mL_1,\mL_3$ and $\mL_5$; $a-b$ on $\mL_2$; $-a+b$ on $\mL_4$ and $\mL_6$.

As was mentioned above, all the points of the trajectories in the sets of degeneration of the forms induced by the symplectic structure are degenerate critical points of the momentum map $\mF$ for any rank. The formal proof is given in \cite{RyKh2012}. There are no such points among the points of rank 0. Let us discuss degenerate periodic trajectories.

Let $\tau_i=\mD_i \cap Q_1$. It is known (see \cite{Zot}) that $\tau_1$ and $\tau_2$ consist of one trajectory each, and $\tau_3$ consists of two trajectories. The images of the sets $\tau_i$ are the above introduced points $e_i$ ($i=1,2,3$), so $e_i=\Delta_1 \cap \delta_i$. In the image of $\mF_2$ these points lie on the axis $m=0$. The same trajectories are the intersections $\mD_i \cap \mn$ and the union of them coincides with the intersection of $Q_1$ with $\mo$. The sets $\mD_i$ have no common points with $Q_2$. The only non-empty intersection with $Q_3$ exists for $\mD_3$ and consists of the pair of trajectories $\tau_4$, corresponding to the above mentioned value $s_0$ which is the root of equation \eqref{eq3_21}.

The intersections of $Q_i$ with $\mL_j$ are as follows:
\begin{itemize}
\item $Q_1$ intersects $\mL_3,\mL_4,\mL_2$ by the above mentioned trajectories $\tau_1,\tau_2,\tau_3$ and does not have any other intersections with $\mL_j$;

\item $Q_2$ intersects $\mL_4$ by the pair of rotational trajectories $\tau_5$ with $h=\frac{a^2+3b^2}{2b}$ and $\mL_2$ by the pair of rotational trajectories $\tau_6$ with $h=\frac{3a^2+b^2}{2a}$; there are no other intersections of $Q_2$ with $\mL_j$;

\item the trajectories $\tau_5$ and $\tau_6$ also serve as the intersections of $Q_3$ with $\mL_4$ and $\mL_2$ respectively; moreover $Q_3$ intersects $\mL_5$ by one oscillation type trajectory $\tau_7$ ($h=-2\sqrt{ab}$) and the pair of rotational trajectories $\tau_8$ ($h=2\sqrt{ab}$).
\end{itemize}

To understand the whole picture of the position of the families $\mL_i$ in the critical subsystems, it is useful to mention the following facts.
\begin{propos}\label{prop1}
$1.$ The subsystem $\mm$ does not have common motions with $\mL_1,\mL_5,\mL_6$ and contains the trajectories $\tau_1,\tau_2,\tau_3$ of the families $\mL_3,\mL_4,\mL_2$. These trajectories also lie in $Q_1$ and serve as its intersections with the families $\mD_1,\mD_2,\mD_3$.

$2.$ The families $\mL_1,\ldots,\mL_4$ lie completely in the intersection $\mn \cap \mo$.

$3.$ The family $\mL_5$ intersects the subsystem $\mn$ by the bifurcational trajectories of the levels $h=\pm(a+b)$; the trajectories of the family $\mL_5$ lie in $\mo$ for all values $h \notin (-2\sqrt{ab},2\sqrt{ab})$. The boundary values of $h$ correspond to the trajectories $\tau_7,\tau_8$.

$4.$ The family $\mL_6$ intersects with the subsystem $\mn$ by the bifurcational trajectories of the levels $h=\pm(a-b)$ and completely lies in $\mo$.

$5.$ The subsystem $\ml$ is the union of the families $\mL_5$ and $\mL_6$; the part of $\ml$ not belonging to other subsystems consists of the trajectories of $\mL_5$ with the values $h \in (-2\sqrt{ab},2\sqrt{ab})$.
\end{propos}

Let us denote the images of the families $\mL_i$ under the momentum maps by $\lambda_i$ ($i=1,\ldots,6$). These are one-dimensional objects. The image points of the special trajectories $\tau_i$ will be denoted by $e_i$ ($i=1,\ldots,8$).

Let us now collect the information on the images of the above mentioned special points and subsets of the phase space.

{\renewcommand{\arraystretch}{1.5}
\begin{table}[!htbp]
\centering
\begin{tabular}{|c|c|c|c|c|c|c|}
\multicolumn{7}{r}{\fts{Table \myt\label{tabee}}}\\
\hline
${}$ & ${h}$ & ${k}$ & ${g}$ & ${f}$ & ${m}$ & ${s}$ \\
\hline
${e_1}$ & ${-2b}$ & ${0}$ & ${-b p^2}$ & ${0}$ & ${0}$ & ${-b}$ \\
\hline
${e_2}$ & ${2b}$ & ${0}$ & ${bp^2}$ & ${0}$ & ${0}$ & ${b}$ \\
\hline
${e_3}$ & ${2a}$ & ${0}$ & ${ap^2}$ & ${0}$ & ${0}$ & ${a}$ \\
\hline
${e_4}$ & ${h(s_0)}$ & ${0}$ & ${g(s_0)}$ & ${f(s_0)}$ & ${-}$ & ${s_0}$ \\
\hline
${e_5}$ & ${\ds \frac{a^2+3b^2}{2b}}$ & ${\ds \frac{r^4}{4 b^2}}$ & ${\ds \frac{p^4-b^2 r^2}{2b}}$ & ${-}$ & ${-\ds \frac{1}{2b}}$ & ${b}$ \\
\hline
${e_6}$ & ${\ds \frac{3a^2+b^2}{2a}}$ & ${\ds \frac{r^4}{4a^2}}$ & ${\ds \frac{p^4+a^2 r^2}{2a}}$ & ${-}$ & ${-\ds \frac{1}{2a}}$ & ${a}$ \\
\hline
${e_7}$ & ${-2\sqrt{ab}}$ & ${(a-b)^2}$ & ${-2(ab)^{3/2}}$ & ${-}$ & ${-}$ & ${-\sqrt{ab}}$ \\
\hline
${e_8}$ & ${2\sqrt{ab}}$ & ${(a-b)^2}$ & ${2(ab)^{3/2}}$ & ${-}$ & ${-}$ & ${\sqrt{ab}}$ \\
\hline
${e_9}$ & ${p\sqrt{2}}$ & ${\ds \frac{r^4}{2p^2}}$ & ${\ds \frac{2p^4+r^4}{2p\sqrt{2}}}$ & ${-}$ & ${-\ds \frac{1}{p\sqrt{2}}}$ & $\ds{\frac{p}{\sqrt{2}}}$ \\
\hline
\end{tabular}\,
\end{table}
}

In Table~\ref{tabee}, we give the values of all first integrals at the points of the trajectories $\tau_i$, i.e., the coordinates of the points $e_i$ in all image spaces of the momentum maps.  Here for the point $e_4$ the coordinates $h,g$ and $f$ are calculated by the formulas \eqref{eq3_20} for $\delta_3$ with the above defined value $s_0$.

\begin{remark}\label{rem2}
In Table~$\ref{tabee}$, we add a new point $e_9$. It corresponds to the minimal value of the energy $H$ on the set of degenerate critical points $Q_2$. Unlike the sets $Q_1$ and $Q_3$, where all extremal energy values provide the appearance of degenerate periodic orbits, i.e., correspond to bifurcations of degenerate two-dimensional tori, on the set $Q_2$ the minimal value of $H$ is reached without any additional bifurcations of the degenerate critical tori of rank 2. Nevertheless this value is essential for the classification of the iso-energy invariants.
\end{remark}

The images of the sets $Q_i$ of degenerate critical points of rank 1 and 2 under the map $\mF$ have the form
\begin{equation*}
\begin{array}{l}
  \Delta_1: \left\{\begin{array}{l} k=0, \\
2g=p^2h
\end{array}\right., \quad h \geqslant -2b.
 \\
  \Delta_2: \left\{\begin{array}{l}
k=\ds{\frac{1}{r^4}(2g-p^2h)^2} \\
g=g_{\pm}(h)=\ds{\frac{1}{4p^2}} \left[(2p^4-r^4)h \pm r^4 \sqrt{\mstrut h^2-2p^2} \right]
\end{array}\right., \quad h \geqslant \sqrt{\mstrut 2p^2}. \\
\Delta_3: \left\{\begin{array}{l}
h=\ds{\frac{a^2b^2+3 s^4}{2 s^3}}\\
k=\ds{-\frac{3s^2}{4}+ a^2+b^2 -\frac{3a^2b^2}{2 s^2}+\frac{a^4b^4}{4 s^6}} \\
g=\ds{\frac{3a^2b^2+s^4}{2s}}
\end{array}\right., \quad 0 < s \leqslant s_0.
\end{array}
\end{equation*}
The equations for the sets $\Delta_i$ in the image planes of the partial momentum maps are given in Table~\ref{tabdel}.

{
\begin{table}[!htbp]
\centering
\begin{tabular}{|c|c|c|c|}
\multicolumn{4}{r}{\fts{Table \myt\label{tabdel}}}\\
\hline
{} & \begin{tabular}{c}In ${\bR}^2(f^2,h)$ \end{tabular} &
\begin{tabular}{c}In ${\bR}^2(m,h)$\end{tabular} &
\begin{tabular}{c}In ${\bR}^2(s,h)$\end{tabular}\\
\hline
$\Delta_1$ & \begin{tabular}{l}$f=0$,\; $h\geqslant -2b$\end{tabular} &
\begin{tabular}{l}$m=0$,\; $h\geqslant -2b$ \end{tabular} &
$e_1(-b,-2b),e_2(b,2b),e_3(a,2a)$\\
\hline
$\Delta_2$ & -- & \begin{tabular}{l}$2p^2m^2+2h m+1 =0,$\;$m<0$\end{tabular}&
\begin{tabular}{l}$2s^2-2hs+p^2=0,$\;$s>0$\end{tabular}\\
\hline
$\Delta_3$ & -- & -- &
\begin{tabular}{l}$3s^4-2hs^3+a^2b^2=0,$\;$s\in(0,s_0]$\end{tabular}\\
\hline
\end{tabular}
\end{table}
}

On the families ${\mL}_j$, the values of the first integrals fill the curves $\lambda_j$. In  $\bR^3(h,k,g)$, the equations of these curves are
\begin{equation*}
\begin{array}{l}
\lambda_{1,2}=\{g = a^2 h\pm a r^2,k=(h\pm
2a)^2,h \geqslant \mp(a\pm b)\},\\
\lambda_{3,4}=\{g = b^2 h\mp b r^2, k=(h\pm
2b)^2, h \geqslant -(a\pm b)\},\\
\lambda_{5,6}=\{g=\pm abh,k=(a\mp b)^2,h\geqslant -(a\pm b)\}.
\end{array}
\end{equation*}
On the plane $\bR^2(m,h)$ the images of $\mL_1,\ldots,\mL_4$ under $\mF_2$ are the half-lines
\begin{equation*}
\begin{array}{l}
\lambda_{1,2}=\{h = r^2 m \mp 2a, h \geqslant \mp(a\pm b)\},\qquad
\lambda_{3,4}=\{h =-r^2 m \mp 2b, h \geqslant -(a\pm b)\}.
\end{array}
\end{equation*}
On the plane $\bR^2(s,h)$ the images of all six families under $\mF_3$ are
\begin{equation*}
\begin{array}{ll}
\lambda_{1,2}=\{s = \mp a, h \geqslant \mp(a\pm b)\}, &
\lambda_{3,4}=\{s = \mp b, h \geqslant -(a\pm b)\},\\[1mm]
\lambda_{5}=\{h=s + \ds \frac{ab}{s}, s\in [-a,-b] \cup (0,+\infty)\}, &
\lambda_{6}=\{h=s - \ds \frac{ab}{s}, s\in [-a,0) \cup [b,+\infty)\},\\
\end{array}
\end{equation*}
In all these formulas the upper sign corresponds to the family with the first number.

Collecting the images on the planes under the partial momentum maps $\mF_1,\mF_2,\mF_3$ of the specific sets of critical points described above, we obtain the bifurcation diagrams $\Sigma_1^*, \Sigma_2^*, \Sigma_3^*$ of the subsystems $\mm,\mn,\mo$ respectively. They are shown in Fig.~\ref{fig_bifset1} -- \ref{fig_bifset3}. Let us give some comments on these figures.
\begin{enumerate}
 \item The typical inverse image of the points on the curves $\Delta_1$ and $\Delta_2$ in the critical set, as it was already mentioned, consists of degenerate critical points of rank 2, which normally are not bifurcational inside a critical subsystem with two degrees of freedom. Nevertheless, with this choice of the partial integrals, $\Delta_1$ and $\Delta_2$ become the outer boundaries of the regions of the existence of motions for the subsystems $\mm$ and $\mn$, therefore they are the parts of the bifurcation diagrams. The phenomena taking place in the neighborhood of the pre-images of these curves inside the subsystems $\mm$ and $\mn$ are in details studied in \cite{Zot1,KhGEOPHY}. In the third subsystem $\mo$ the curves $\Delta_2,\Delta_3$ do not cause any bifurcations with respect to $\mo$. But they are the images of degenerate points of the complete momentum map and thus separate the points of $\mo$ having different outer type. Thus, these sets are also shown in Fig.~\ref{fig_bifset3} (with dashed lines).

 \item The diagrams $\Sigma_i^*$ ($i=1,2,3$) split the points of rank 2 regular with respect to the subsystems $\mm,\mn$ and $\mo$ into definite classes (for $\mo$ we add to the separating set also the curves $\Delta_2,\Delta_3$). These classes (to be more exact, the subregions arising in the image of the partial momentum map and corresponding to non-empty integral manifolds) are denoted for the first three subsystems by the symbols, respectively, $a,b,c$ supplied with indices. The connected components of the supplement to bifurcation diagrams are usually called \textit{chambers} \cite{BolFom}. Here for the partition into chambers we also take into account the images of the 2-tori degenerate in $\mP^6$.

 \item The diagram $\Sigma_4^*$ consists of the isolated energy values separating different types of periodic solutions. These values are $h=\pm a\pm b$ at the points of rank 0 and $h=\pm 2\sqrt{ab}$; the latter separate the motions in the family $\mL_5$ belonging to to the subsystem $\mo$ and isolated from another critical subsystems. For the family $\mL_5$ the image segment which is not shown in the figures and correspond to the values $h\in(-2\sqrt{ab},2\sqrt{ab})$ is denoted by $\lambda_{50}$.

\end{enumerate}

\begin{remark}\label{remsegm}
The curves $\Delta_i$, $\lambda_j$ are divided into different segments by the node points $e_1,\ldots,e_9$. A similar division on the curve $\delta_3$ is generated by the point $e_4$. We will mark such segments by putting a second index after the curve number.
\end{remark}

With the above information, let us give an explicit description of the bifurcation diagram $\Sigma$ of the map $\mF$. Introduce some notation.

On the curves \eqref{eq3_20}, we define the inversions of the dependencies $h(s)$ on monotonous segments:
\begin{equation*}
\begin{array}{llll}
\delta_1:& s=s_1(h), & h\in[-2b,+\infty), & s_1(h)\in
[-b,0),\\
\delta_2: & s=s_2(h), & h \geqslant 2b, & s_2(h)\in
(0,b],\\
\delta_{31}: & s=s_{31}(h), & h \in [h_0,2a], & s \in [a, s_0],\\
\delta_{32}: & s=s_{32}(h), & h \in [h_0,+\infty), & s \in
[s_0+\infty).
\end{array}
\end{equation*}
Here $h_0$ is the value $h(s_0)$ on the curve $\delta_3$. We now obtain the equation \eqref{eq3_21} for the value $s_0$ from the condition for the minimum of $h$ on the curve $\delta_3$ in the following form
\begin{equation}\label{eq3_27}
\frac{1}{s}\sqrt{(s^2-a^2)(s^2-b^2)}=\sqrt{s^2-b^2}-\sqrt{s^2-a^2}.
\end{equation}
Here the fact that it has a unique solution at $s>a$ is obvious.

From equations \eqref{eq3_9} for $\pov_1$ we find the dependency on~$\delta_1$:
$$
g = g_1(h) = s^3+\frac{ab}{s}-s^2 \phi(s)|_{s=s_1(h)}, \quad h
\geqslant -2b.
$$
Considering the intervals where $h(s)$ is monotonous on the curves $\lambda_5 -
\lambda_6$, we denote
$$
\displaystyle{s_{51}(h)= \frac{h - \sqrt{h^2-4ab}}{2},}\quad
\displaystyle{s_{52}(h)= \frac{h + \sqrt{h^2-4ab}}{2},}\quad
\displaystyle{s_6(h)= \frac{h + \sqrt{h^2+4ab}}{2}}.
$$
Now the bifurcation diagram is completely described by the following theorem \cite{Kh36}, which is formulated in such a way that all the conditions for the parameters on the surfaces are explicit inequalities when the energy value $h$ is fixed.

\begin{theorem}\label{theo2}
$1.$ The set $\Sigma_1=\pov_1 \cap \Sigma$ has the form
\begin{equation*}
\left\{\begin{array}{l}
         h \geqslant -2b \\
         k = 0,\quad  g_1(h) \leqslant g \leqslant \frac{1}{2}p^2h
       \end{array}
   \right. .
\end{equation*}
$2.$ The set $\Sigma_2=\pov_2 \cap \Sigma$ lies in the half-space $h \geqslant -(a+b)$ and is described by the following collection of the systems of inequalities:
\begin{equation*}
\left\{
\begin{array}{l}
-(a+b)\leqslant h \leqslant p\sqrt{2}\\
b^2 h -b r^2\leqslant g \leqslant a^2 h +a r^2
\end{array}
\right.; \quad
\left\{
\begin{array}{l}
h \geqslant p\sqrt{2}\\
b^2 h -b r^2 \leqslant g \leqslant g_{-}(h)
\end{array}
\right.; \quad
\left\{
\begin{array}{l}
h \geqslant p\sqrt{2}\\
g_{+}(h) \leqslant g \leqslant a^2 h +a r^2
\end{array} \right. .
\end{equation*}
$3.$ The set $\Sigma_3=\pov_3 \cap \Sigma$ is completely described by the following collection of conditions in the $(s,h)$-plane. For the negative values of $s$
\begin{equation*}
\left\{
\begin{array}{l} -(a+b)\leqslant h \leqslant -2 \sqrt{ab}\\
s \in [-a,s_{51}(h)] \cup [s_{52}(h),-b]
\end{array}
\right.; \quad
\left\{
\begin{array}{l} -2 \sqrt{ab}\leqslant h \leqslant -2b\\
s \in [-a,-b]
\end{array}
\right.; \quad
\left\{
\begin{array}{l} h >  -2b\\
s \in [-a,s_1(h)]
\end{array}
\right..
\end{equation*}
For the positive values of $s$
\begin{equation*}
\begin{array}{ll}
\left\{
\begin{array}{l} -a+b\leqslant h \leqslant 2b\\
s \in [b,s_6(h)]\end{array}
\right.;
&
\left\{
\begin{array}{l} 2b \leqslant h \leqslant h_0\\
s \in [s_2(h),s_6(h)]
\end{array}
\right.; \\[4mm]
\left\{
\begin{array}{l} h_0 \leqslant h \leqslant 2a\\
s \in [s_2(h),s_{31}(h)]\cup[s_{32}(h),s_6(h)]
\end{array}
\right.;
&
\left\{
\begin{array}{l} h > 2a\\
s \in [s_2(h),a]\cup[s_{32}(h),s_6(h)]
\end{array}
\right..
\end{array}
\end{equation*}
$4.$ The set $\Sigma_4=\pov_4 \cap \Sigma$ consists of two half-lines
\begin{equation*}
\left\{
\begin{array}{l}
h \geqslant  - (a + b)\\
g = abh,\quad k = (a - b)^2
\end{array}
\right.;
\quad
\left\{
\begin{array}{l}
h \geqslant  - a + b\\
g =  - abh,\quad  k = (a + b)^2
\end{array}
\right. .
\end{equation*}
\end{theorem}

\begin{remark}\label{remtheo2}
This theorem provides a possibility to present in details any cross section of the diagram $\Sigma$ by the planes of fixed $h$ and to trace with the computer graphics the evolution of these cross section with the energy change. The boundary conditions for the value of $h$ are separating values of the energy. In particular, the conditions for $\Sigma_2$ make it clear why the point $e_4$ was introduced as the extremal value of $h$ on $\Delta_2$ $($see Remark~$\ref{rem2})$. Nevertheless, the boundary values of $h$ do not necessarily form the complete separating set. More exact statement will be given below.
\end{remark}

%\clearpage

\section{Classification of critical points by type}\label{sec4}
All necessary definitions dealing with the notion of a non-degenerate critical point and the type of a critical point are given in \cite{BolFom}. Let us give one comment on the terminology used here.

In the topological analysis of the systems with two degrees of freedom $\cite{BolFom}$ the following terms are used: a 3-atom to describe bifurcations in the neighborhood of a non-degenerate point of rank $1$ (it is a connected component of the pre-image of a small segment transversal to a smooth segment of the bifurcation diagram; this component is foliated into Liouville 2-tori with one singular fiber) and a 4-atom to describe a saturated neighborhood of a non-degenerate point of rank $0$ (usually in terms of an almost direct product of some atoms from the systems with one degree of freedom).
\begin{remark}\label{remt}
Dealing with a system with three degrees of freedom we fix the term 3-atom to characterize a bifurcation in the neighborhood of a point of rank $1$ \emph{inside} the corresponding critical subsystem with two degrees of freedom, while a 4-atom will always mean a bifurcation in the complete six-dimensional space in the neighborhood of a non-degenerate point of rank $2$.
\end{remark}

Thus, we define a 4-atom as a foliated into Liouville 3-tori with one singular fiber connected component of the pre-image of a small segment transversal to a smooth leave of the bifurcation diagram of the momentum map $\mF:\mP^6\to \bR^3$. Studying the iso-energy diagrams we may consider such a small segment belonging to the corresponding fixed energy level, since at any non-degenerate point of rank 2 the Hamiltonian is a regular function.

At four points of rank 0 (singular points of the initial system) the Hamiltonian $H$ is a Morse function and, as it was already mentioned, ${\rm ind}\,H(p_i)=i$. This fact almost completely determines the character of the system in the neighborhood of these points. Nevertheless, the strict classification demands establishing the types of them as the critical points of the momentum map.
\begin{theorem}\label{theo3}
All critical points of rank $0$ are \textit{non-degenerate} in $\mP^6$. Moreover, $\imm_0$ has the type ``center-center-center'', $\imm_1$ has the type ``center-center-saddle'', $\imm_2$ has the type ``center-saddle-saddle'', and $\imm_3$ has the type ``saddle-saddle-saddle''.
\end{theorem}

{\renewcommand{\arraystretch}{1.5} \setlength{\extrarowheight}{-2pt}
\begin{table}[!ht]
\centering
\small
\begin{tabular}{|c|c|c|c|c|}
\multicolumn{5}{r}{\fts{Table \myt\label{tab41}}}\\
\hline $\mK^0$
&
\begin{tabular}{c}Image in\\ ${\bR}^3(h,k,g)$ \end{tabular}
& Type in $\mn$
& \begin{tabular}{c}Image in\\ ${\bR}^2(s,h)$ \end{tabular}
& Type in $\mo$ \\
\hline
$\imm_0$ &$P_0$ &center-center & \begin{tabular}{l}$P^-_{0}$\\$P^+_{0}$\end{tabular}
& \begin{tabular}{l} center-center\\ center-center \end{tabular}
\\
\hline
$\imm_1$ &$P_1$ &center-saddle
& \begin{tabular}{l} $P^-_{1}$\\$P^+_{1}$ \end{tabular}
& \begin{tabular}{l} center-saddle\\ center-center \end{tabular}
\\
\hline
$\imm_2$ &$P_2$ &center-saddle
& \begin{tabular}{l} $P^-_{2}$\\$P^+_{2}$ \end{tabular}
& \begin{tabular}{l} saddle-saddle\\ center-saddle \end{tabular}
\\
\hline
$\imm_3$ & $P_3$ &saddle-saddle
& \begin{tabular}{l} $P^-_{3}$\\$P^+_{3}$ \end{tabular}&
\begin{tabular}{l} saddle-saddle\\ saddle-saddle \end{tabular}
\\
\hline
\end{tabular}
\end{table}
}

For the proof, in \cite{RyKh2012} the explicit form of the characteristic polynomials of the symplectic operator $A_H$ are calculated at the points $\imm_i$.

Note that at the points $\imm_i$, three local critical subsystems meet, namely, the subsystem $\mn$ and the two parts of the subsystem $\mo$ (recall that this subsystem has at these points a singularity of the self-intersection type). In particular, on the $(s,h)$-plane each point $\imm_i$ is represented by two points. In this sense, each point of the pair $P^{\pm}_{i}$ of the diagram $\Sigma_3^*$ has its own type (the type of the point $\imm_i$ with respect to some chosen smooth part of $\mo$ in the neighborhood of $\imm_i$). The corresponding description of the ctical points of rank 0 in the critical subsystems $\mn$ and $\mo$ is given in Table~\ref{tab41}.

%%%%%%%%%%%%%%%%%%%%%%%%%%%%%%%%%%%%%%%%%%%%%%%%%%%%%%%
%%%%%%%%%%%%%%%%%%%%%%%%%%%%%%%%%%%%%%%%%%%%%%%%%%%%%%%%%%%
%\vspace{5mm}

Now let us turn to the points of rank 1 organized into special periodic trajectories.

\begin{theorem}\label{theo4}
The points of rank $1$ forming the families of trajectories $\mL_i$ и $\mD_j$ ${(i=1,\ldots,6;}$ $j=1,2,3)$ are non-degenerate as the singularities of the map $\mF$ except for the following values of the energy:

{\parskip=2mm

on $\mD_1$: $h = -2b$;

on $\mD_2$: $h = 2b$;

on $\mD_3$: $h = 2a$, $h=h(s_0)$;

on $\mL_2$: $h = 2a$, $h={\frac{3a^2+b^2}{2a}}$;

on $\mL_3$: $h = -2b$;

on $\mL_4$: $h = 2b$, $h={\frac{a^2+3b^2}{2b}}$;

on $\mL_5$: $h =\pm 2\sqrt{a b}$.

}

Depending on the family and the energy value, the type of non-degenerate singularities in $\mP^6$ and the corresponding 3-atoms in the critical subsystems are given in Table~$\ref{tab42}$.
\end{theorem}

%%%%%%%%%%%%%%%%%%%%%%%

%%%%%%%%%%%%%%%%%%%%%%%%%%%%%%%%%%%%%%%%%%%%%%%%%%%%
{\renewcommand{\arraystretch}{1.5} \setlength{\extrarowheight}{0pt}
\centering
\small
\tabcolsep=3pt
\begin{longtable}{|c|l|c|c|c|c|c|}
\multicolumn{7}{r}{\fts{Table \myt\label{tab42}}}\\
\hline
$\mK^1$
& \multicolumn{1}{c|}{Image or segment}
& {\renewcommand{\arraystretch}{0.8}\begin{tabular}{c}Number of \\ trajectories \end{tabular}}
& {Type in $\mP^6$}
& {\renewcommand{\arraystretch}{0.8}\begin{tabular}{c}3-atom\\in $\mm$\end{tabular}}
& {\renewcommand{\arraystretch}{0.8}\begin{tabular}{c}3-atom\\in $\mn$\end{tabular}}
& {\renewcommand{\arraystretch}{0.8}\begin{tabular}{c}3-atom\\in $\mo$\end{tabular}}
\\
\hline\endfirsthead%
\multicolumn{7}{r}{\fts{Table \ref{tab42} (continued)}}\\
\hline
$\mK^1$
& \multicolumn{1}{c|}{Image or segment}
& {\renewcommand{\arraystretch}{0.8}\begin{tabular}{c}Number of \\ trajectories \end{tabular}}
& {Type in $\mP^6$}
& {\renewcommand{\arraystretch}{0.8}\begin{tabular}{c}3-atom\\in $\mm$\end{tabular}}
& {\renewcommand{\arraystretch}{0.8}\begin{tabular}{c}3-atom\\in $\mn$\end{tabular}}
& {\renewcommand{\arraystretch}{0.8}\begin{tabular}{c}3-atom\\in $\mo$\end{tabular}}
\\
\hline\endhead

\multirow{2}{*}{${\mL}_1$}
&
$\lambda_{11}:-(a+b)<h<-(a-b)$
&
{1}
&
center-center
&{}&{$A$}&{$A$}
\\
\hhline{|~|-|-|-|-|-|-|}
{}
&
$\lambda_{12}:h>-(a-b)$
&
{2}
&
center-center
&{}&{$2A$}&{$2A$}\\
\hline
\multirow{3}{*}{${\mL}_2$}
&
$\lambda_{21}:a-b<h<a+b$
&
{1}
&
saddle-saddle
&{}&{$B$}&{$B$}\\
\hhline{|~|-|-|-|-|-|-|}
{}
&
$\lambda_{22}:a+b<h<\frac{3a^2+b^2}{2a}$
&
{2}
&
saddle-saddle
&{}&{$2B$}&{$2B$}\\
\hhline{|~|-|-|-|-|-|-|}
{}
&
$\lambda_{23}:\frac{3a^2+b^2}{2a}<h<2a$
&
{2}
&
saddle-saddle
&{}&{$2B$}&{$2B$}\\
\hhline{|~|-|-|-|-|-|-|}
{}
&
$\lambda_{24}: h>2a$
&
{2}
&
saddle-center
&{}&{$2B$}&{$2A$}\\
\hline
\multirow{3}{*}{${\mL}_3$}
&
$\lambda_{31}:-(a+b)<h<-2b$
&
{1}
&
center-center
&{}&{$A$}&{$A$}\\
\hhline{|~|-|-|-|-|-|-|}
{}
&
$\lambda_{32}:-2b<h<a-b$
&
{1}
&
center-saddle
&{}&{$A$}&{$B$}\\
\hhline{|~|-|-|-|-|-|-|}
{}
&
$\lambda_{33}: h>a-b$
&
{2}
&
center-saddle
&{}&{$2A$}&{$C_2$}\\
\hline
\multirow{3}{*}{${\mL}_4$}
&
$\lambda_{41}:-(a-b)<h<2b$
&
{1}
&
saddle-center
&{}&{$B$}&{$A$}\\
\hhline{|~|-|-|-|-|-|-|}
{}
&
$\lambda_{42}:2b<h<a+b$
&
{1}
&
saddle-saddle
&{}&{$B$}&{$B$}\\
\hhline{|~|-|-|-|-|-|-|}
{}
&
$\lambda_{43}: a+b<h <\frac{a^2+3b^2}{2b}$
&
{2}
&
saddle-saddle
&{}&{$2B$}&{$C_2$}\\
\hhline{|~|-|-|-|-|-|-|}
{}
&
$\lambda_{44}: h >\frac{a^2+3b^2}{2b}$
&
{2}
&
saddle-saddle
&{}&{$2B$}&{$C_2$}\\
\hline

\multirow{4}{*}{${\mL}_5$}
&
$\lambda_{50}:-2\sqrt{ab}<h<2\sqrt{ab}$
&
{1}
&
focus-focus
&{}&{}&{}
\\
\hhline{|~|-|-|-|-|-|-|}

{}
&
$\lambda_{51}:-(a+b)<h<-2\sqrt{ab}$
&
{1}
&
center-center
&{}&{}&{$A,A$}\\
\hhline{|~|-|-|-|-|-|-|}
{}
&
$\lambda_{52}:2\sqrt{ab}<h<a+b$
&
{1}
&
saddle-saddle
&{}&{}&{$B,B$}\\
\hhline{|~|-|-|-|-|-|-|}
{}
&
$\lambda_{53}:h>a+b$
&
{2}
&
saddle-saddle
&{}&{}&{$2A^*,2B$}\\
\hline

\multirow{2}{*}{${\mL}_6$}
&
$\lambda_{61}:-(a-b)<h<a-b$
&
{1}
&
saddle-center
&{}&{}&{$B,A$}\\
\hhline{|~|-|-|-|-|-|-|}
{}
&
$\lambda_{62}:h>a-b$
&
{2}
&
saddle-center
&{}&{}&{$2A^*,2A$}\\
\hline

${\mD}_1$
&
$\delta_1:s\in(-b,0)$
&
{2}
&
center-center
&{$2A$}&{}&{$2A$}\\
\hline

${\mD}_2$
&
$\delta_2:s\in(0,b)$
&
{2}
&
saddle-center
&{$2B$}&{}&{$2A$}\\
\hline

\multirow{2}{*}{${\mD}_3$}
&
$\delta_{31}: s\in(a,s_0)$
&
{4}
&
saddle-center
&{$4B$}&{}&{$4A$}\\
\hhline{|~|-|-|-|-|-|-|}
{}
&
$\delta_{32}: s\in(s_0,+\infty)$
&
{4}
&
center-center
&{$4A$}&{}&{$4A$}\\
\hline
\end{longtable}
}

Note that the cases of degeneration correspond to the points $e_1 - e_8$ on the bifurcation diagrams.

Let us comment the order of words characterizing the type of points with respect to $\mP^6$ in Table~\ref{tab42}. For the families $\mL_j$ with $i=1,\ldots,4$ the first type is the type of the point of rank 1 in the subsystem $\mn$, and the second type is the type of this point in the subsystem $\mo$. For the families $\mD_{1,2,3}$ the first type is the type of the point of rank 1 in the subsystem $\mm$, and the second type is the type of this point in the subsystem $\mo$.
The families $\mL_{5,6}$ (except for the pre-image of the focus type segment $\lambda_{50}$, which is isolated in the space of the integral constants, and for its boundary points with degenerate rank 1 critical points in the pre-image) form the set of the transversal self-intersection of the phase space of the subsystem $\mo$. This implies that each trajectory corresponds to the pair of values $s$ and to two points $\lambda^{\pm}_{ij}$ ($i=5,6$) in the diagram of the subsystem $\mo$ (see Fig.~\ref{fig_bifset3}). The minus sign corresponds to the smaller value of~$s$. Therefore, in Table~\ref{tab42} for these families the first type is the type of the point of rank 1 in the part of the subsystem $\mo$ with smaller $s$, and the second type is the type of this point in the part of the subsystem $\mo$ with larger $s$.

For the proof of the theorem, in the work \cite{RyKh2012}, the pairs of the first integrals generating the Cartan subalgebra of dimension 2 in the algebra of the symplectic operators are explicitly written out and it is shown that the following integrals generate regular elements in these subalgebras:
\begin{equation*}
\begin{array}{ll}
\textrm{на}\quad{\mL}_{1,2}:& \vp_{1,2}=K-2(h\pm 2a)H, \\
\textrm{на}\quad{\mL}_{3,4}:& \vp_{3,4}=K-2(h\pm 2b)H, \\
\textrm{на}\quad{\mL}_{5,6}:& \vp_{5,6}=\pm abH-G, \\
\textrm{на}\quad{\mD}_{i}: & \psi_i=2G-(p^2-\vk_i)H\qquad (i=1,2,3),
\end{array}
\end{equation*}
where
\begin{equation*}
\begin{array}{ll}
\displaystyle{\vk_1=(\sqrt{a^2-s^2}+\sqrt{b^2-s^2})^2}, & s \in
[-b,0); \\
\displaystyle{\vk_2=(\sqrt{a^2-s^2}-\sqrt{b^2-s^2})^2}, & s \in
(0,b]; \\
\displaystyle{\vk_3=-(\sqrt{s^2-b^2}-\sqrt{s^2-a^2})^2}, & s \in
[a,+\infty).
\end{array}
\end{equation*}
The atoms of all bifurcations taking place inside the critical subsystems $\mm,\mn,\mo$ at the periodic trajectories consisting of the rank 1 points are established as a result of the topological analysis of these subsystems (see \cite{Zot,KhSav,Kh2009}).

Let us consider the set $\mK^2$ of the critical points of rank 2. It is the union of three subsystems $\mm$, $\mn$ and $\mo$, from which we remove already investigated points $\mK^0 \cup \mK^1$. According to dimensions the type of a non-degenerate rank 2 singularity can only be \textit{elliptic} (``center'') or \textit{hyperbolic} (``saddle''). Since the property to be non-degenerate is applied to the whole Liouville torus and then all the points of this torus have the same type, we speak about elliptic or hyperbolic tori.

First of all note that any integral 2-torus regular in the subsystem $\mm$ and non-degenerate in $\mP^6$ is elliptic. Indeed, the integral $K=Z_1^2+Z_2^2$ is everywhere non-negative function which vanishes exactly on $\mm$.

\begin{theorem}\label{theo5} All critical points of rank 2 on the manifold $\mm$, except for the points of the zero level of the integral $F$, are non-degenerate of the elliptic type.
\end{theorem}
\begin{proof}
On $\mm$, there are no points of dependence of the integrals $H$ and $G$ except for the sets $\mD_1 - \mD_3$ \cite{Kh2005} (note that the dependence of the \textit{restrictions} $H|_{\mm}$ and $F|_{\mm}$ or, which is the same, of the \textit{restrictions} $H|_{\mm}$ and $G|_{\mm}$ is investigated in \cite{Zot,Zot1}). At the same time everywhere on $\mm$ we have $dK=0$. The characteristic equation of the operator $A_K$ in $\bR^9$ is easily written out with the help of the equations of the manifold $\mm$; it has seven zero roots and the remaining factor
$\mu^2+4 f^2$ has two different imaginary roots for $f\ne 0$. The theorem is proved.
\end{proof}

In the image of the momentum maps the value $F=0$ gives the set $\Delta_1 \subset \bR^3(h,k,g)$. Thus, for this set we have proved that the corresponding critical points are degenerate.

%%%%%%%%%%%%%%%%%%%%%%%%%%%%%%%%%%%%%%%%%%%%%%%%%
On the manifold $\mn$, system \eqref{eq2_1} has an explicit algebraic solution \cite{KhSav}. This means that all phase variables are expressed as rational functions of two auxiliary variables $s_1,s_2$ and some radicals of the form $\sqrt{s_i - c_j}$, where $c_j$ are constant values depending on the constants of the integrals $H,M$. In the variables $s_1,s_2$, the equations of motion separate and reduce to elliptic quadratures. The corresponding formulas can also be found in \cite{RyKh2012}.

\begin{theorem}\label{theo6} All critical points of rank 2 on the manifold $\mn$, except for the points in the pre-image of the curves $\Delta_1$ and $\Delta_2$, are non-degenerate. They have elliptic type for $m>0$ and hyperbolic type for $m<0$.
\end{theorem}
\begin{proof}
Due to quite simple equations \eqref{eq3_9} of the leave $\pov_2$, it is convenient to take as a single integral playing the role of \eqref{eq2_11} and having a singularity at each point of $\mn \cap \mK^2$ the function
\begin{equation}\label{eq4_5}
\Phi = \phi_2(H,K,G)=(2G-p^2H)^2-r^4 K.
\end{equation}
For this function, after the necessary factorization by the zero kernel space and substitution of the explicit expressions for the phase variables via the separation variables,
the characteristic equation of the operator $A_\Phi$ takes the form
\begin{equation}\label{eq4_6}
\mu^2+4r^{12} \ell^2 m=0,
\end{equation}
and, except for the cases $m=0$ ($\Delta_1$) and $\ell=0$ ($\Delta_2$),
has two different roots, purely imaginary for $m>0$ and real for $m<0$.
Here $\ell$ is the constant of the partial integral $L$ defined in \eqref{eqelu}.
The theorem is proved.
\end{proof}

As it is shown in \cite{KhSav}, $\mn$ consists of critical points of the zero level of the functions $2G -p^2H \pm r^2 \sqrt{K}$, one of which is the sum of the squares of two regular smooth functions ($m>0$), and the other one is the difference of such squares ($m<0, \ell \ne 0$). This is the reason why the tori are elliptic in the first case and hyperbolic in the second one. Equation \eqref{eq4_6} strictly proves the non-degeneracy of such points. It is interesting to note the connection between degenerate points and the analytic solution. When $m=0$, the degree of the polynomial standing under the radical in the separated equations of the Abel\,--\,Jacobi type jumps from four to three. When $\ell=0$, the variables $s_1,s_2$ in the separated equations have only even powers, which provides an additional symmetry. As a whole, $\mn$ is non-orientable and the main role here plays the neighborhood of the set ${\ell=0}$ \cite{KhGEOPHY}.

%%%%%%%%%%%%%%%%%%%%%%%%
On the manifold $\mo$, the explicit algebraic solution of system \eqref{eq2_1} is given in \cite{Kh2009,Kh2007}.
All phase variables are expressed as rational functions of two auxiliary variables $t_1,t_2$ and some radicals of the form $\sqrt{t_i - c_j}$, where $c_j$ are constant values depending on the constants of the integrals $H,S$. In the variables $t_1,t_2$, the equations of motion separate and have the form of the Abel\,--\,Jacobi equations or the equations of Kovalevskaya with a polynomial of degree six under the radical. This separation is hyperelliptic. The corresponding formulas also are given in~\cite{RyKh2012}.

%%%%%%%%%%%%%%%%%%%%%%%%%%%%%%%%%%%%%%%%%%%%%%%
%%%%%%%%%%%%%%%%%%%%%%%%%%%%%%%%%%%%%%%%%%%%%%%%%%%%%
\begin{theorem}\label{theo7} All regular two-dimensional tori of the subsystem $\mo$
consist of non-degenerate critical points of rank $2$ of the momentum map ${\mF}$, except for the points in the pre-images of the sets $\Delta_2, \Delta_3$. The torus is elliptic if the value
$
s(2s^2-2h s+a^2+b^2)[s^4+2(s-h)s^3+a^2b^2]
$
is negative and hyperbolic if it is positive.
\end{theorem}
\begin{proof}
As an integral having singularity on $\mo$ one can take the function of the type \eqref{eq4_5} obtained, for example, by eliminating $s$ in equations \eqref{eq3_9} for the surface $\pov_3$. Nevertheless, the result is too huge and this approach is not convenient. It is better to consider the function with Lagrange multipliers introduced in \cite{Kh2005} for finding the equations of critical subsystems, $\Psi = 2 G + 2s(s-h) H+s K$, where $s$ and $h$ are first considered as undefined multipliers. As shown in \cite{Kh2005}, {\it after} writing down the condition that $\Psi$ has a critical point the constants $s,h$ on $\mo$ turn out to be the values of the integrals $S,H$. Thus, calculating the characteristic polynomial of the operator $A_\Psi$ we consider $s$ and $h$ to be some constants, and after that we substitute in the obtained expression the dependencies of the phase variables on the separation variables. Then we get the characteristic equation
\begin{equation*}
\mu^2-\frac{2(2s^2-2h s+a^2+b^2)}{s}\Bigl[s^4+2(s-h)s^3+a^2b^2\Bigr]=0.
\end{equation*}
This proves the theorem.
\end{proof}

Finally, let us again turn to Fig.~\ref{fig_bifset1} -- \ref{fig_bifset3}, on which the images of the critical points of rank 2  fill the regions $a_1 - a_3$, $b_1-b_9$, and $c_1 - c_{17}$. For the critical subsystems these regions are in fact the chambers, since we include in $\Sigma_i^*$ also the image of the degenerate 2-tori. According to Theorems~\ref{theo5}\,--\,\ref{theo7} we establish the type of all points of rank 2.

\begin{propos}\label{propr2}
The critical points of rank 2 have the elliptic type in the pre-images of the regions $a_1-a_3$, $b_1-b_3$, $c_1 - c_4,c_{10},c_{13},c_{17}$ and the hyperbolic type in the pre-images of the rest of the regions. In the pre-images with elliptic bifurcations the 4-atoms are as follows:

$A$ for the regions $c_1,b_1$;

$2A$ for the regions $a_1,b_2,c_2,c_3,c_4,c_{10},c_{13}$;

$4A$ for the regions $b_4,b_5,b_6,c_{12},c_{16} $;

$8A$ for the region $a_3$.

\noindent In the pre-images of the sets $\Delta_1,\Delta_2$, and $\Delta_3$ the two-dimensional tori consist of degenerate points of rank~2.
\end{propos}

The information on the number of critical 2-tori in the pre-images of all mentioned regions is obtained during the rough topological analysis of the critical subsystems in the works \cite{Zot1,KhSav,Kh2009} and, for the sake of brevity, it is shown in Table~\ref{tabatom} of the next section. It is clear that the elliptic 4-atoms by this information are totally defined. For the hyperbolic atoms the necessary foundations will be given later.

Thus, we gave a total classification of the critical points of the momentum map with respect to their types and 3-atoms, and this completes the description of the rough phase topology of the critical subsystems.

%\clearpage

\section{Iso-energy atlas}\label{sec5}
Iso-energy diagram $\sigh$ is the bifurcation diagram of the restriction of the momentum map $\mF$ to the iso-energy level $H_h=\{H=h\}\subset \mP^6$. The natural identification is used of this restriction with the map to the plane of the constants $g,k$:
\begin{equation*}
  \mF|_{\{H=h\}} \cong (G{\times}K)|_{\{H=h\}}: H_h \to \bR^2(g,k).
\end{equation*}
For brevity, the last map is denoted by $\mF(h)$.

We now consider the problem of classification of the diagrams $\sigh$ equipped by additional information on the topology of the pre-image, i.e., we also point out the number of the families of regular tori in the supplement to the diagram and the atoms of those bifurcations that take place when crossing the 1-dimensional skeleton of the diagram. The diagram $\sigh$ itself is considered as a stratified manifold. Smooth segments of the one-dimensional skeleton are the images of non-degenerate points of rank 2 and zero-dimensional skeleton is the image of all critical points of rank 1 and degenerate critical points of rank 2. Here we speak, of course, of the points on a fixed regular iso-energy level $H_h$. Having all previous information and the information revealed here, one can easily endow the points of the zero-dimensional skeleton with the description of the saturated neighborhood of the pre-image. Due to the lack of space, we restrict ourselves to pointing out in the next section the rough loop molecules of hyperbolic rank 0 singularities.

Generally speaking, the diagram can have a very complicated structure. For this reason until now there is no strict and unique definition of the diagrams equivalence. In our case, $\sigh$ is a plane one-dimensional stratified manifold. Therefore we call two diagrams equivalent if there exists a diffeomorphism of their neighborhoods in the plane preserving taking one diagram to another together with the above introduced equipment.

What does it mean to point out the bifurcation? Let $\gamma$ be a smooth segment of the 1-skeleton of $\sigh$ and $z=(g,k)\in \gamma$. Consider a small one-dimensional segment $\eps$ drawn through $z$, transversal to $\gamma$ and having no common points with $\sigh$ other than $z$. The pre-image of $\eps$ in the phase space is a four-dimensional manifold $\mF^{-1}(\eps)\subset H_h$ foliated into Liouville 3-tori with one singular fiber on each connected component. The set of critical points on a singular fiber consists of a finite number of 2-tori, all points of which have rank 2. We call a connected component of $\mF^{-1}(\eps)$ a \mbox{4-a}tom (see Remark~\ref{remt}). Some examples of 4-atoms are the products of the standard 3-atoms with a circle. For theses atoms we keep the generally accepted notation $A,B,C_2,A^*$ etc. Thus, any smooth edge of the equipped diagram must have the notation of the corresponding 4-atom. On the other hand, any smooth edge of $\sigh$ is a cross section at the level $h$ of the corresponding chamber of one of the critical subsystems and the 4-atom is obtained when we consider the bifurcation taking place while crossing transversally the corresponding smooth leaf that is the image of the critical subsystem. So it is convenient to show on the edges of the iso-energy diagrams the notation of the corresponding chambers of the critical subsystems, and then, collecting the obtained information and analyzing the evolution of the diagrams $\sigh$, to point out all 4-atoms corresponding to those chambers.

Obviously $\sigh$ depends not only on $h$ but also on the physical parameters $a,b$. The point $(a,b,h)$ will be called separating if any its neighborhood in $\bR^3$ contains points with non-equivalent diagrams. Since the stratified manifold $\sigh$ is built from the $h$-sections $\Sigma_i(h)$ of the sets $\Sigma_i$ $(i=1,\ldots,4)$, the transformations of $\sigh$ take place if and only if one of the stratified manifolds $\Sigma_i(h)$ is transformed. It follows from Theorem~\ref{theo2} (see Remark~\ref{remtheo2}) that the cross sections $\Sigma_i(h)$ are transformed when $h$ crosses the boundary values met in this theorem. This gives eight separating values of $h$ (depending on $a,b$), i.e., eight separating surfaces in $\bR^3(a,b,h)$. In general, the separating values of $h$ are the critical values of $h$ considered as a function on the stratified manifolds $\Sigma_i$ $(i=1,\ldots,4)$ \cite{KhGDIS,KhVVMSH13}. One can easily see from Fig.~\ref{fig_bifset1} -- \ref{fig_bifset3} that the critical values of $h$ on $\Sigma_i$ are exactly the $h$-coordinates of the points $P_0 - P_3$ and $e_1 - e_9$. Looking at the Table~\ref{tabee} we come to the following theorem~\cite{Kh36}.

\begin{theorem}\label{theo8}
In the space of the parameters $\bR^3(a,b,h)$ there exists $13$ separating surfaces for the classification of the iso-energy diagrams $\sigh$. These surfaces are considered under the natural restriction $0\leqslant b \leqslant a$ and have the form
\begin{equation}\label{eq5_1}
\begin{array}{l}
\Q_1-\Q_4 :\, h =  \mp a \mp b;  \\
\Q_5:\,h=-2b; \qquad \Q_6:\,h=2b; \qquad \Q_7:\,h=2a; \\
\Q_{8} :\, \left\{
\begin{array}{l}
\ds h = s\left(3 - \frac{s^2}{a^2}\right) - \frac{1}{a^2}\sqrt{(s^2  - a^2)^3} \\[3mm]
\ds b  = \frac{s}{a}\sqrt{s \left[s\left(3 - \frac{2s^2}{a^2}\right) -
\frac{2}{a^2}\sqrt{(s^2 - a^2)^3} \right]}
\end{array}
\right.,   \quad s \in [a,\frac{2}{\sqrt{3}}a];\\
\ds \Q_9 :\,h = \frac{1}{2b}(a^2 + 3b^2); \qquad \ds \Q_{10}:\,h = \frac{1}{2a}(3a^2 + b^2); \\
\Q_{11}:\, h = - 2\sqrt {ab}; \qquad \Q_{12}:\, h = 2\sqrt {ab}; \qquad \ds \Q_{13} :\,h = \sqrt{2(a^2+b^2)}.
\end{array}
\end{equation}
Non-empty diagrams exist for the region $h \geqslant -a-b$. All cross sections by the planes of constant $a \ne 0$ are taken to the cross section $a=1$ by the homothetic transformation $h'=h/a$, $b'=b/a$. It corresponds to the fact that $a$ can be taken as a measurement unit for the intensities of the force fields.
\end{theorem}

Here the comments are needed only for the surface $\Q_{8}$. It is obtained as a parametric form of representation of the value $h(s_0)$ at the point $e_4$ by using equation \eqref{eq3_27} and the corresponding equation for $h$ on the curve $\delta_3$.

The separating set is shown in Fig.~\ref{fig_chart} under the projection onto the cross section $a=1$. Iso-energy diagrams are non-empty in 19 regions in the parameters space. Let us choose one point in each region and connect them with paths to obtain convenient ways to visualize and analyze the diagrams transformations. According to Fig.~\ref{fig_chart}, let us call these paths ``the left circle'', ``the right circle'', ``the line'' and ``the block''. The circles have common region 1, the right circle and the line meet at the region 9, and the line and the block have common region 14. The corresponding diagrams are shown in Fig.~\ref{fig_leftc1} -- \ref{fig_block1}. To obtain a more clear picture, we made some changes by plane diffeomorphisms, i.e., in the same equivalence class as defined above. All elements of the diagrams are equipped with notation of the corresponding elements introduced above for the diagrams of the critical subsystems. Moreover,  the roman numbers $\tx{I}$ -- $\tx{IX}$ enumerate the connected components of the supplement to the bifurcation set of the complete momentum map in the admissible domain (i.e., the domain corresponding to non-empty integral manifolds) in the enhanced space $\bR^5(g,k,h,a,b)$. We keep for such components the term \textit{chamber}. In each chamber the structure of the regular integral manifold stays the same for any change of the integral constants and of the physical parameters. The ``external'' chamber, which is not admissible (it contains unbounded values of the integrals and thus obviously corresponds to empty integral manifolds), will be denoted by $\varnothing$ and called the zero-chamber. We can now establish the number of regular 3-tori in the chambers, describe the families of these tori and find all types of the 4-atoms.

\begin{theorem}\label{theon}
Iso-energy diagrams divide the enhanced space of the parameters into $10$ chambers, one of which $($external$)$ corresponds to empty integral manifolds. In the phase space according to all changes of the parameters there exists exactly $23$ families of three-dimensional regular Liouville tori. The chambers are separated by $29$ walls corresponding to the two-dimensional chambers defined by the diagrams of the critical subsystems. While crossing these walls we have the bifurcations defined by the 4-atoms given in Table~$\ref{tabatom}$.
\end{theorem}

{\renewcommand{\arraystretch}{1.5} \setlength{\extrarowheight}{0pt}
\tabcolsep=3pt
\centering
\small
\begin{longtable}{|c|c|c|c||c|c|c|c|}
\multicolumn{8}{r}{\fts{Table \myt\label{tabatom}}}\\
\hline
{\renewcommand{\arraystretch}{0.8}
\begin{tabular}{c} Segment\\(region)\end{tabular}}
&
{\renewcommand{\arraystretch}{0.8}
\begin{tabular}{c} Number of\\2-tori \end{tabular}}
&
Transfer
&
4-atom
&
{\renewcommand{\arraystretch}{0.8}
\begin{tabular}{c} Segment\\(region)\end{tabular}}
&
{\renewcommand{\arraystretch}{0.8}
\begin{tabular}{c} Number of\\2-tori \end{tabular}}
&
Transfer
&
4-atom
\\
\hline\endfirsthead%
\multicolumn{8}{r}{\fts{Таблица \ref{tabatom} (продолжение)}}\\
\hline
{\renewcommand{\arraystretch}{0.8}
\begin{tabular}{c} Segment\\(region)\end{tabular}}
&
{\renewcommand{\arraystretch}{0.8}
\begin{tabular}{c} Number of\\2-tori \end{tabular}}
&
Transfer
&
4-atom
&
{\renewcommand{\arraystretch}{0.8}
\begin{tabular}{c} Segment\\(region)\end{tabular}}
&
{\renewcommand{\arraystretch}{0.8}
\begin{tabular}{c} Number of\\2-tori \end{tabular}}
&
Transfer
&
4-atom
\\
\hline\endhead

$a_1$ & 2 & $\varnothing \to \tx{II}$ & $2A$ & $c_4$ & 2 & $\varnothing \to \tx{IV}$ & $2A$\\
\hline
$a_2$ & 4 & $\varnothing \to \tx{V}$ & $4A$ & $c_5$ & 2 & $\tx{II} \to \tx{V}$ & $2B$\\
\hline
$a_3$ & 8 & $\varnothing \to \tx{VIII}$ & $8A$ & $c_6$ & 1 & $\tx{I} \to \tx{III}$ & $B$\\
\hline
$b_1$ & 1 & $\varnothing \to \tx{I}$ & $A$ & $c_7$ & 2 & $\tx{II} \to \tx{IV}$ & $2A^*$\\
\hline
$b_2$ & 2 & $\varnothing \to \tx{III}$ & $2A$ & $c_8$ & 2 & $\tx{IV} \to \tx{VII}$ & $2B$\\
\hline
$b_3$ & 4 & $\varnothing \to \tx{VII}$ & $4A$ & $c_9$ & 2 & $\tx{III} \to \tx{VI}$ & $2B$\\
\hline
$b_4$ & 4 & $\tx{VII} \to \tx{VIII}$ & $4B$ & $c_{10}$ & 2 & $\tx{III} \to \tx{VI}$ & $2A$\\
\hline
$b_5$ & 4 & $\tx{VI} \to \tx{VII}$ & $2C_2$ & $c_{11}$ & 2 & $\tx{III} \to \tx{VI}$ & $2B$\\
\hline
$b_6$ & 4 & $\tx{VII} \to \tx{IX}$ & $2C_2$ & $c_{12}$ & 4 & $\tx{V} \to \tx{VII}$ & $4A^*$\\
\hline
$b_7$ & 2 & $\tx{III} \to \tx{V}$ & $2B$ & $c_{13}$ & 2 & $\tx{IV} \to \tx{IX}$ & $2A$\\
\hline
$b_8$ & 1 & $\tx{I} \to \tx{II}$ & $B$ & $c_{14}$ & 2 & $\tx{IV} \to \tx{IX}$ & $2B$\\
\hline
$b_9$ & 2 & $\tx{III} \to \tx{IV}$ & $C_2$ & $c_{15}$ & 2 & $\tx{III} \to \tx{VII}$ & $2B$\\
\hline
$c_1$ & 1 & $\varnothing \to \tx{I}$ & $A$ & $c_{16}$ & 4 & $\tx{V} \to \tx{VIII}$ & $4B$\\
\hline
$c_2$ & 2 & $\varnothing \to \tx{II}$ & $2A$ & $c_{17}$ & 4 & $\tx{V} \to \tx{VIII}$ & $4A$\\
\hline
$c_3$ & 2 & $\varnothing \to \tx{III}$ & $2A$ & {} & {} & {} & {}\\
\hline
\end{longtable}
}

The proof needs several steps.

First, using only known elliptic bifurcation we find the number of regular 3-tori in the chambers and the corresponding families.

Recall that a family of regular tori is a foliation to the Liouville tori of a connected component of a set in the phase space which is obtained by removing all singular fibers. The families are fixed inside one chamber, but the same family may be present in several chambers.

Walls between chambers (including walls between a regular chamber and the zero-chamber) which correspond to the 4-atoms of the type $kA$ ($k\in \mathbb{N}$) will be called $A$-crossings. Let us put all $A$-crossings on the diagrams using the bifurcations of the type ``center'' from Proposition~\ref{propr2}. Then we obtain the information collected in Table~\ref{tab51}. Thus, we establish the number of regular tori in all chambers and the total number of families of tori in all chambers. As we see, there are 23 families. The common families are present in the pairs of the chambers $\tx{III},\tx{VI}$ (two families), $\tx{IV},\tx{IX}$ (two families) and $\tx{V},\tx{VIII}$ (four families).

{\renewcommand{\arraystretch}{1.5} \setlength{\extrarowheight}{0pt}
\begin{table}[!htbp]
\centering
\small
\begin{tabular}{|c|c|c|c|c|c|}
\multicolumn{6}{r}{\fts{Table \myt\label{tab51}}}\\
\hline
Chamber &
{\renewcommand{\arraystretch}{0.8}
\begin{tabular}{c} Crossing \end{tabular}}
& Segments & 4-atom &
{\renewcommand{\arraystretch}{0.8}
\begin{tabular}{c} Number\\of tori \end{tabular}} &
{\renewcommand{\arraystretch}{0.8}
\begin{tabular}{c} New\\families\end{tabular}}
\\
\hline
$\tx{I}$ & $\varnothing \to \tx{I} $ & $c_1,b_1$ & $A$ & {1} &{1}\\
\hline
$\tx{II}$ & $\varnothing \to \tx{II} $ & $a_1,c_2$ & $2A$ & {2} &{2}\\
\hline
$\tx{III}$ & $\varnothing \to \tx{II} $ & $b_2,c_3$ & $2A$ & {2} &{2}\\
\hline
$\tx{IV}$ & $\varnothing \to \tx{IV} $ & $c_4$ & $2A$ & {2} &{2}\\
\hline
$\tx{V}$ & $\varnothing \to \tx{V} $ & $a_2$ & $4A$ & {4} &{4}\\
\hline
$\tx{VI}$ & $\tx{III} \to \tx{VI} $ & $c_{10}$ & $2A$ & {4} &{2}\\
\hline
$\tx{VII}$ & $\varnothing \to \tx{VII} $ & $b_3$ & $4A$ & {4} &{4}\\
\hline
$\tx{VIII}$ & $\tx{V} \to \tx{VIII} $ & $c_{17}$ & $4A$ & {8} &{4}\\
\hline
$\tx{IX}$ & $\tx{IV} \to \tx{IX} $ & $c_{13}$ & $2A$ & {4} &{2}\\
\hline
\end{tabular}
\end{table}
}

The following steps are dealing with defining all hyperbolic bifurcations. Some of them can be found by the next lemma, the analog of which for two degrees of freedom can be obtained from the results of~\cite{BolFom}.

\begin{lemma}\label{lem1}
Let $x_0$ be a non-degenerate critical point of rank $n-2$ of the type ``center-saddle'' in an integrable system with $n$ degrees of freedom and the momentum map $\mathcal{J}$. Denote by $U_1$ and $U_2$ two local critical subsystems with $n-1$ degrees of freedom transversaly intersecting at $x_0$ and such that in $U_1$ the point $x_0$ has some hyperbolic $n$-atom $V$ and in $U_2$ the $n$-atom of $x_0$ is the elliptic atom $kA$ $(k\in \mathbb{N})$. Let us consider a small two-dimensional section at the point $\mathcal{J}(x_0)$ transversal to the smooth leaves $\mathcal{J}(U_1)$ and $\mathcal{J}(U_2)$. Let $q$ be a point on the edge of this section generated by the leaf $\mathcal{J}(U_2)$. Then the $(n+1)$-atom containing the pre-image of $q$ is $k(V{\times}S^1)$.
\end{lemma}

In Table~\ref{tab52}, we show all the edges of the iso-energy diagrams satisfying the conditions of the lemma, i.e., the edges along which inside the corresponding critical subsystem in the pre-image the boundary bifurcation is elliptic, but the pre-image of the boundary point has the type ``center-saddle''. The hyperbolic 3-atom at the boundary point is found from Table~\ref{tab42}, then, on the edges considered, the 4-atoms are defined from the lemma (they are given in the last column of the table). Recall that the 4-atoms which are direct products of 3-atoms and the circle have the same notation as the initial 3-atoms.

{\renewcommand{\arraystretch}{1.5} \setlength{\extrarowheight}{0pt}
\begin{table}[!htbp]
\centering
\small
\begin{tabular}{|c|c|c|c|c|}
\multicolumn{5}{r}{\fts{Table \myt\label{tab52}}}\\
\hline
Edge &
{\renewcommand{\arraystretch}{0.8}
\begin{tabular}{c} Boundary\\point \end{tabular}}
&
{\renewcommand{\arraystretch}{0.8}
\begin{tabular}{c} Boundary\\crossing \end{tabular}}
&
{\renewcommand{\arraystretch}{0.8}
\begin{tabular}{c} 3-atom \end{tabular}}
&
{\renewcommand{\arraystretch}{0.8}
\begin{tabular}{c} 4-atom\\of the edge\end{tabular}} \\
\hline
$b_8$ & $\lambda_{32}$ & $c_2 \to c_1$ & $B$ в $\mo$& $B$ \\
\hline
$c_6$ & $\lambda_{41}$ & $b_1 \to b_2$ & $B$ в $\mn$ & $B$ \\
\hline
$b_9$ & $\lambda_{33}$ & $c_4 \to c_3$ & $C_2$ в $\mo$& $C_2$ \\
\hline
$c_7$ & $\lambda_{62}^{-}$ & $c_2 \to c_4$ & $2A^*$ в $\mo$& $2A^*$ \\
\hline
$c_5$ & $\delta_{2}$ & $a_1 \to a_2$ & $2B$ в $\mm$ & $2B$ \\
\hline
$c_{16}$ & $\delta_{31}$ & $a_2 \to a_3$ & $4B$ в $\mm$ & $4B$ \\
\hline
\end{tabular}
\end{table}
}

Now let us consider the cusps $\Delta_{31} - \Delta_{33}$ of the iso-energy diagrams. Since for theses points we know the 4-atoms on the incoming elliptic edges ($2A$ on the edges $c_{13}$ at $\Delta_{31}$ and $c_{10}$ at $\Delta_{32},\Delta_{33}$ ), we see that the remaining incoming edges $c_{14},c_9$ and $c_{11}$ have the atoms $2B$. The cusp $\Delta_{34}$ does not giva any new information (the atoms on the edges $c_{16}$ and $c_{17}$ are already known).

One more block of information is obtained from the following simple reasoning. Let the critical points in the pre-image of some edge form $k$ hyperbolic two-dimensional tori, and crossing the edge between the adjacent chambers changes the number of regular tori from $k$ to $2k$. Then the atom on the edge is $kB$. From here, having the information about 2-tori, we establish the following atoms: $2B$ on the edges $b_7$, $c_8$ and $c_{15}$ (crossings $\tx{III}\to \tx{V}$, $\tx{IV}\to \tx{VI}$ and $\tx{III}\to \tx{VII}$ respectively), $4B$ on the edge $b_4$ (crossing $\tx{VII}\to \tx{VIII}$).

Still we have unknown 4-atoms on the edges $b_5,b_6$ and $c_{12}$. All of them correspond to the crossings between the chambers containing four tori. Let us analyze the information we already have on the non-degenerate periodic trajectories in the pre-images of the following points of the iso-energy diagrams: $\lambda_{43}$ for $b_5$ (region 12),  $\lambda_{44}$ for $b_6$ (region 15),  and $\lambda_{53}^{\pm}$ for $c_{12}$ (region 13). For all of these points we know the 4-atoms on three incoming edges out of four and also the number of regular tori in the adjacent chambers. From here we find out that the 4-atoms for the edges left are $2C_2$ on the edges $b_5,b_6$ and $4A^*$ on the edge $c_{12}$. The theorem is proved.

This completes the description of the equipped iso-energy diagrams and the rough thre-dimensional topology of the regular iso-energy levels.

%\clearpage

\section{Rough net invariants}\label{sec6}

The theory of the topological invariants of the integrable systems with arbitrary many degrees of freedom was developed by A.T.\,Fomenko in the works \cite{fom91,fomams91}. At the same time, due to the absence of methods of practical analysis of non-trivial examples up to this moment no investigations of a system not reducible to a family of systems with two degrees of freedom were fulfilled. Such a system is the Kovalevskaya top in a double field studied here. The equipped iso-energy diagrams built above show one of the possible rough topological invariant. Nevertheless, as it can be seen from the above illustrations, the attempt to keep real scales lead to the appearance of extra small regions of the diagrams transformations. Then the clear picture of these transformations is not always possible to obtain and the illustrations themselves in some details are overloaded with information.

Here we give the description of the rough topological invariants in the form of the analog of the Fomenko nets on the iso-energy levels for all parametrically stable cases. Let us explain the idea of the Fomenko nets and the simplified analog suggested in this work.

Having an integrable system with three degrees of freedom, let us identify to a point each conneted component of integral manifolds. Under some conditions we obtain a three-dimensional stratified manifold (see~\cite{fom91,Fom1989} for details). The rough topological invariant of the system is a one-dimensional complex built from this manifold. It is a graph with two types of vertices. A  vertex of the type ``center'' replaces each connected component of any three-dimensional stratum. A  vertex of the type ``atom'' stands for each connected component of any two-dimensional wall and is  equipped with the notation of the atom of the bifurcation, which takes place in the neighborhood of  corresponding critical integral surface. A graph edge is drawn from a center vertex to that atom  vertex on which the bifurcation of the corresponding three-dimensional torus takes place when crossing the wall. This invariant is called the net invariant of Fomenko.
This invariant may be naturally applied also to reducible systems (i.e., systems having
a one-dimensional symmetry group and admitting the reduction to a family of systems with two degrees of freedom). The experience of constructing multidimensional invariants in such systems brought us to the following conclusions. First, it is more rational to build an invariant for the systems induced on nonsingular iso-energy manifolds. Second, in order to describe the rough topology it is convenient to simplify the construction basing on the notion of an iso-energy diagram and generated chambers.

Let, as before, $H$ be the Hamilton function of a system, $\mF$ the
momentum map, $\Sigma\subset{\bR}^3$ the bifurcation diagram, i.e.,
the set of critical values of the momentum map. Suppose that all iso-energy manifolds $H_h=\{H=h\}$ are compact and fix a nonsingular $H_h$. The diagram $\sigh$ of the restriction $\mF(h)$ of the map $\mF$ to $H_h$ divides the plane into open connected components (chambers) having regular three-dimensional tori in its pre-image. As it was mentioned above, due to the fact that $H_h$ is compact, the zero-chamber arises containing all sufficiently far points of the plane. Let us add the ``infinitely distant'' point to the plane to obtain the two-dimensional sphere $S^2$. If there are no focus type singularities, i.e., the diagram does not have isolated points, then the diagram $\sigh$ generates a presentation of $S^2$ as a two-dimensional cell complex. Let us denote this complex by $\csigh$. The equipped diagram adds to the complex the information on the number of regular tori in the pre-images of the 2-cells and on the bifurcations which correspond to 1-cells.  
Replace each chamber (including the zero-chamber) with a point supplied with the number of regular tori in the pre-image of any point in the chamber. These points are the vertices of the new graph $\Gamma_h$, which is natural to consider lying on the two-dimensional sphere ${S}^2$.
The vertices corresponding to two adjacent chambers are joined by the edge, on which the set of atoms is marked taking place in the neighborhood of the preimage of a point on the common wall of the chambers. If the boundary of a chamber contains an isolated point of the diagram $\sigh$ which
have a focus type singularity of rank 1 in its pre-image, then at the vertex corresponding to this chamber we draw a loop edge supplied by the symbol F. It is convenient to draw the graph $\Gamma_h$ using the projection onto the plane of the sphere $S^2$ from its ``infinitely distant'' point. The corresponding zero-vertex is represented by a dashed closed curve surrounding all other vertices
and edges of the graph. Note that the $F$-loops does not lie on the sphere, but are only shown on it for simplicity. If there are no such loops (i.e., no focus type singularities), then the sphere together with the graph $\Gamma_h$ is a two-dimensional complex conjugate to the complex $\csigh$.

Now we return to the Kovalevskaya top in a double field. According to the above results on the classification of the iso-energy diagrams $\sigh$ with respect to $h$ and the physical parameters $a,b$ (see Fig.~\ref{fig_chart}) we come to the following statement.

\begin{theorem}\label{theo9}
The parametrically stable diagrams $\sigh$ correspond to $19$ types of the rough topological net invariants $\Gamma_h(a,b)$ of the five-dimensional iso-energy surfaces of the Kovalevskaya top in a double field.
\end{theorem}

Note that the topology of the iso-energy levels is known \cite{KhZot} and it changes when crossing the values $h=\mp a\mp b$, i.e., the separating surfaces $\Q_1 - \Q_4$.

The complete list of the invariants is shown in Fig.~\ref{fig_leftcnet}~--~\ref{fig_blocknet}. Here the circles with numbers inside them stand for the vertices together with the number of regular tori (except for the above described way of presenting the vertex for the outer zero-chamber). On the graph edges we place the notation of the 4-atoms of three-dimensional bifurcations obtained
from the corresponding atoms for a system with two degrees of freedom by multiplying by the topological circle. The illustrations are placed in the same way as the corresponding iso-energy diagrams in Fig.~\ref{fig_leftc1}~--~\ref{fig_block1}.

The obtained isoenergetic invariants give a possibility to construct the complete net invariant of Fomenko in the sense of \cite{fom91} on the six-dimensional phase space. To this end, we must watch the dynamics of $\Gamma_h(a,b)$ along the fixed values of $a,b$ and unite these graphs on $h$ in the corresponding enhanced space.  Recall that $a$ is not an essential parameter and all objects are preserved by the homothetic transformation $h'=h/a,b'=b/a$. The complete invariant changes when the line $b'={\rm const}$ in the plane $(b',h')$ crosses one of the double points of the set of the separating curves. The corresponding values of $b'$ are straightforwardly defined from \eqref{eq5_1}. In particular, the point of intersection of the curves $\Q_8$ and $\Q_9$ is defined by the unique positive root $b_0$ of the equation $6591b'^8+612b'^3-262b'^2-28b'-1=0$.

\begin{theorem}\label{theo10}
The problem of the motion of the Kovalevskaya top in a double field has four stable with respect to the physical parameter $b/a\in [0,1]$ rough net invariants of Fomenko. The separating values of the parameter are $0$ $($the Kovalevskaya case$)$, $3-2\sqrt{2}$,  $\frac{1}{3}$, $b_0$ и $1$  $($the reducible Yehia case$).$
\end{theorem}

Knowing the simplified net invariant, let us establish the number and the character of periodic solutions and degenerate motions. Since $H_h$ is supposed non-singular the zero-skeleton of $\sigh$ is generated by periodic solutions consisting of critical points of rank~$1$.
Let us also include in the zero-skeleton the image of degenerate points of rank $2$ filling two-dimensional tori of special type. The graph $\Gamma_h$ (not counting on the $F$-loops) divides the sphere $S^2$ into open connected components; each component is bounded by a cycle of the graph. Thus the 2-cells of the partition of the sphere generated by the graph $\Gamma_h$ correspond to singularities of rank $1$ or to degenerate singularities of rank $2$. Knowing the atoms along the edges bounding a cell corresponding to periodic solutions, one can define the number of these solutions on a given level of the first integrals and the type of these solutions.
Ignoring the $F$-loops, we obtain that the solutions of the type ``center-center'' have on the boundary of the cell only atoms of the type $A$, and the solutions of the type ``saddle-saddle''  only hyperbolic atoms (here these atoms are $B,A^*$ and $C_2$); for the solutions of the type ``center-saddle'' on the boundary of the cell we meet atoms of both elliptic and hyperbolic types. The number of periodic solutions in the pre-image of a 2-cell (the node point of the diagram) is defined by the multiplicity of atoms on the edges and the multiplicity of vertices of graph $\Gamma_h$ along the bounding cycle of a 2-cell. To distinguish 2-cells corresponding to degenerate solutions of rank $2$, it is necessary to address the corresponding bifurcation diagram and select the node points which are cusps or tangency points of the bifurcation curves.

\begin{theorem}\label{theo11}
The number and the types of singular periodic solutions and the number of degenerate two-dimensional tori on non-singular energy levels of the Kovalevskaya top in a double field in $19$ parametrically stable cases are given in Table~$\ref{tab61}$. The regions are enumerated according to Fig.~$\ref{fig_chart}$.
\end{theorem}

\begin{table}[ht]
\centering
%\footnotesize
\tabcolsep=3pt
\begin{tabular}{|p{40mm}|c|c|c|c|c|c|c|c|c|c|c|c|c|c|c|c|c|c|c|}
\multicolumn{20}{r}{\fts{Table \myt\label{tab61}}}\\
\hline
Region number &\textbf{01}&\textbf{02}&\textbf{03}&\textbf{04}&\textbf{05}&\textbf{06}&\textbf{07}&\textbf{08}&\textbf{09}&\textbf{10}&\textbf{11}&\textbf{12}&\textbf{13}&\textbf{14}&\textbf{15}&\textbf{16}&\textbf{17}&\textbf{18}&\textbf{19}\\
\hline
Type ``center-center''&4&3&4&3&2&3&4&4&4&4&4&4&4&4&4&8&8&8&8\\
\hline
Type ``center-saddle''&3&2&2&0&0&1&4&4&6&6&5&6&6&6&6&10&10&8&8\\
\hline
Type ``saddle-saddle''&0&0&0&0&0&0&1&2&3&2&1&6&6&6&6&6&6&4&4\\
\hline
Type ``focus-focus''&1&1&0&0&1&1&1&0&0&1&1&0&0&0&0&0&0&0&0\\
\hline
Degenerate tori&1&0&0&0&0&1&2&6&6&2&1&6&10&14&12&10&10&12&12\\
\hline
\end{tabular}
\end{table}
%%%%%%%%%%%%%%%%%%%%%%%%%%%%%%%%

From the net invariants built above we also get the description of all loop molecules of rank $1$ singularities. If the type of a singularity includes the part ``center'', then its saturated neighborhood in $\mP^6$ is presented as the direct product of two 3-atoms and is easily found from Table~\ref{tab42}. Having a singularities of the type ``saddle-saddle'', let us consider its loop molecule on the iso-energy level. In the introduced simplified net invariant such a singularity with necessity generates a 2-cell with four vertices and with the edges marked by the atoms of hyperbolic three-dimensional bifurcations. The complete list of the forms of the saddle type loop molecules met in this problem together with their notation is given in Fig.~\ref{fig_circle}. The first three molecules have the analogs in two degrees of freedom; these analogs are the molecules of the critical points of rank $0$ with the saturated neighborhoods $B{\times}B$, $(B{\times}C_2)/\mathbb{Z}_2$ and $B{\times}C_2$ \cite{Osh1,Osh2}. Here, of course, we use the standard notations for the 2-atoms. The last molecule does not have such an analog. For all saddle type singularities in our problem we finally get the following iso-energy loop molecules: $\mathrm{NM}_1$ for $\lambda_{42}$ and $\lambda_{52}$; $2\mathrm{NM}_1$ for $\lambda_{23}$; $\mathrm{NM}_2$ for $\lambda_{21}$; $2\mathrm{NM}_2$ for $\lambda_{22}$; $\mathrm{NM}_3$ for $\lambda_{43}$ and $\lambda_{44}$; $2\mathrm{NM}_4$ for $\lambda_{53}$.

\section*{Благодарности}

The authors thank Academician A.T.\,Fomenko for valuable discussions of the approaches and results. 
The work is partially supported by the RFBR (grant N~14-01-00119) and the Authority of Volgograd Region (грант N~13-01-97025).

%%%%%%%%%%%%%%%%%%%%%%%%%%
\clearpage

%\bibliographystyle{unsrturl}
%\bibliography{FAM}

\clearpage
\begin{figure}[!htbp]
\centering
\includegraphics[width=0.25\textwidth,keepaspectratio]{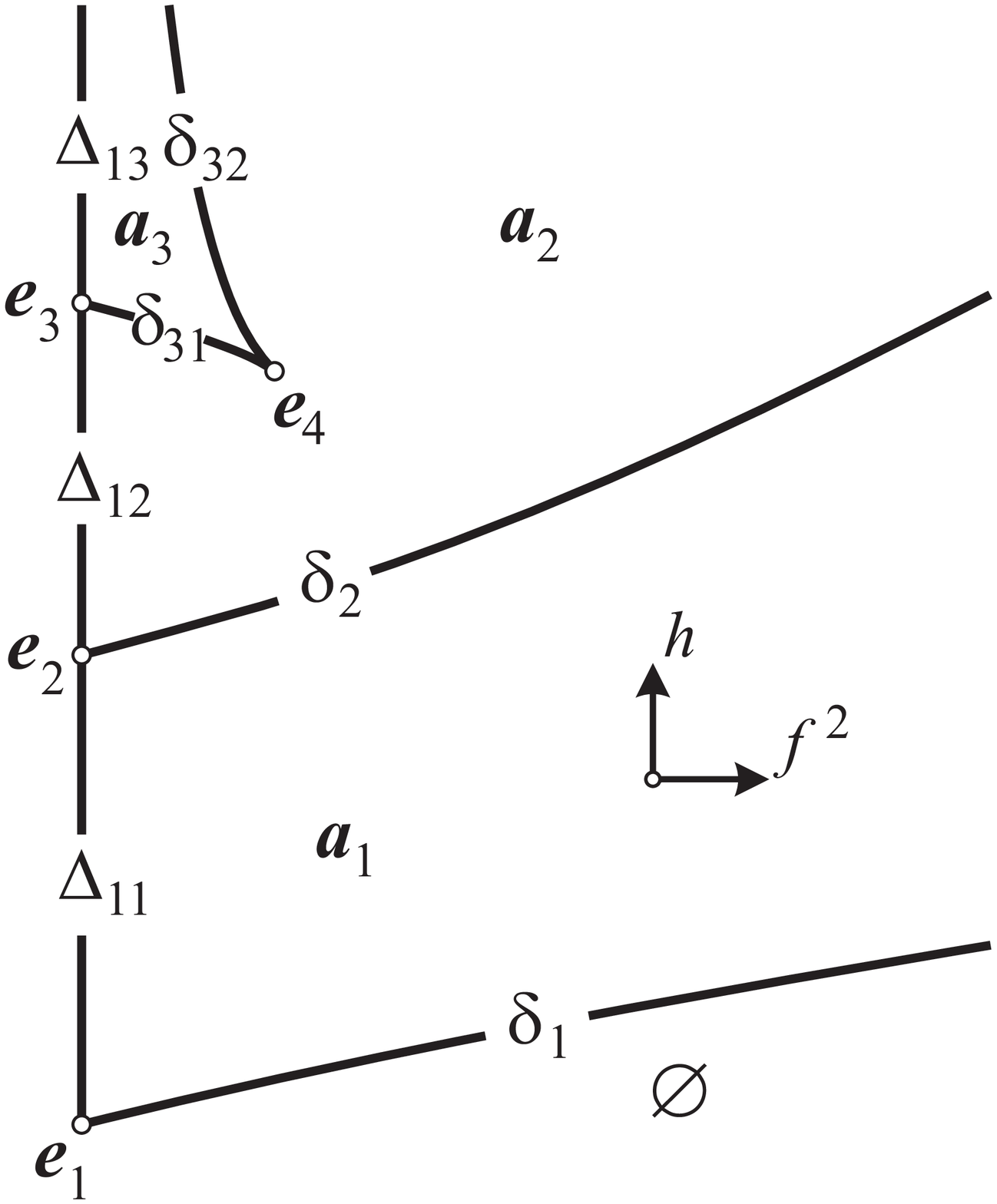}
\caption{The diagram of the subsystem $\mm$ on the $(f^2,h)$-plane.}\label{fig_bifset1}
\end{figure}

\begin{figure}[!htbp]
\centering
\includegraphics[width=0.3\textwidth,keepaspectratio]{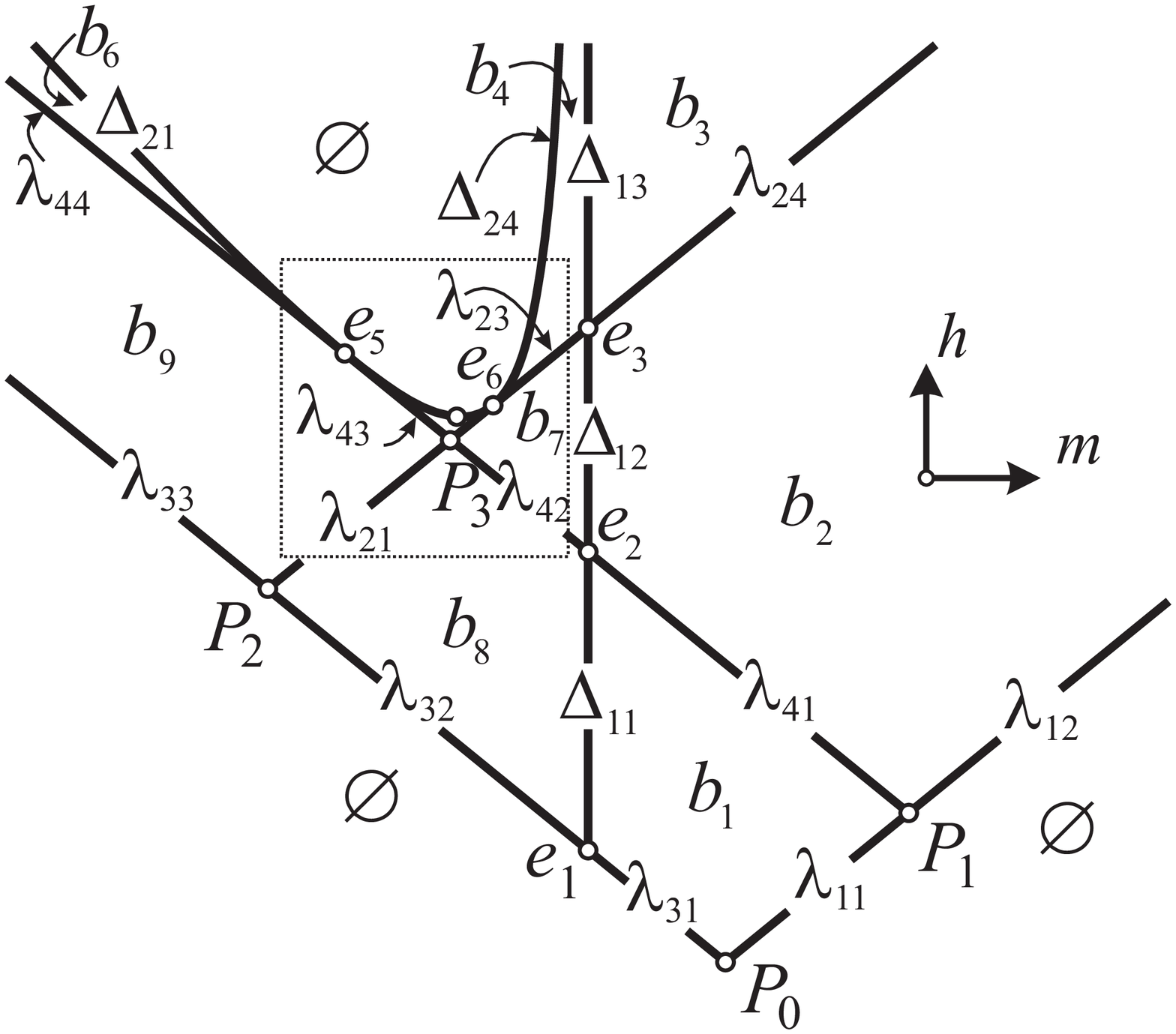}\ \includegraphics[width=0.22\textwidth,keepaspectratio]{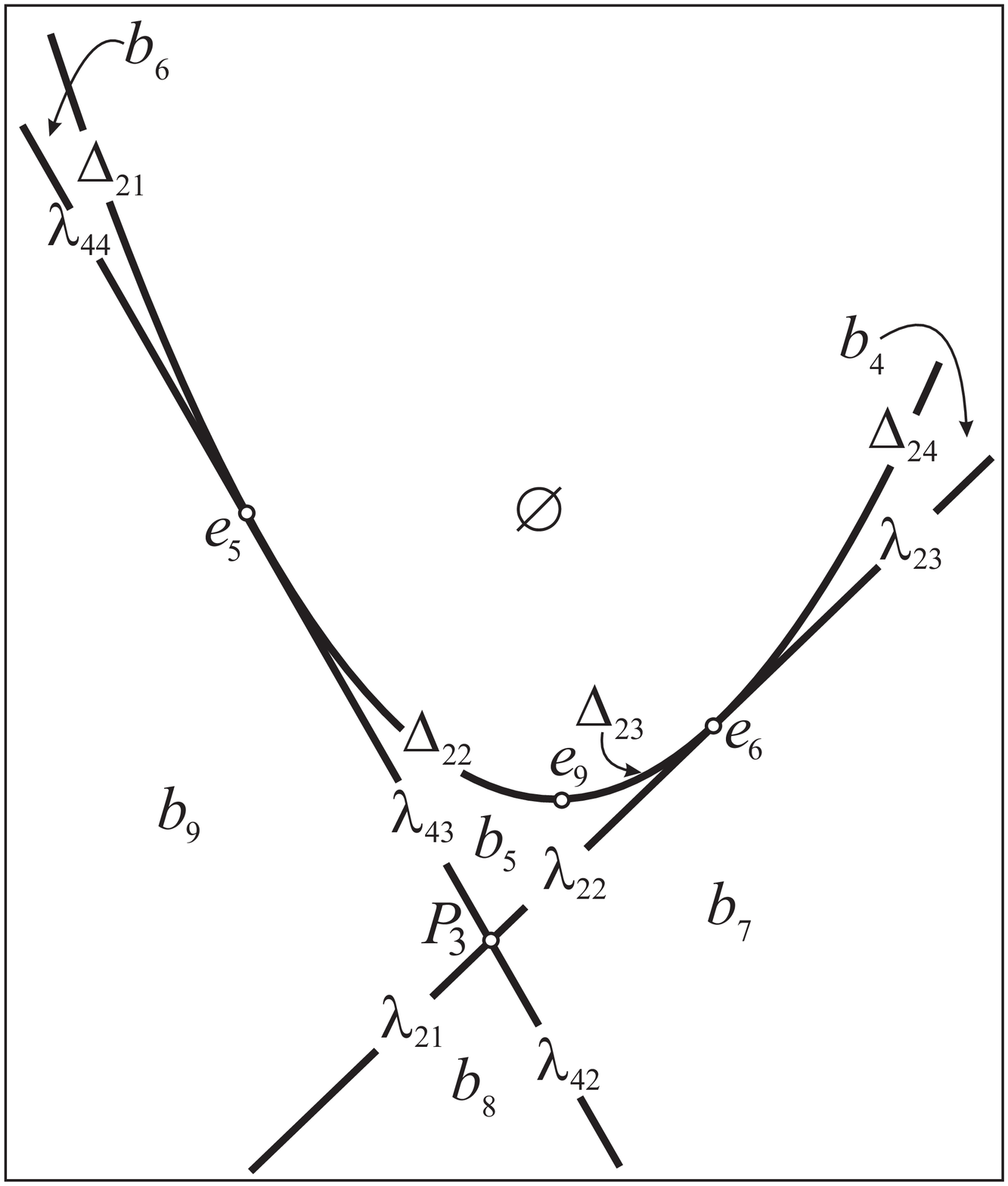}
\caption{The diagram of the subsystem $\mn$ on the $(m,h)$-plane and its fragment.}\label{fig_bifset2}
\end{figure}

\begin{figure}[!htbp]
\centering
\includegraphics[width=0.41\textwidth,keepaspectratio]{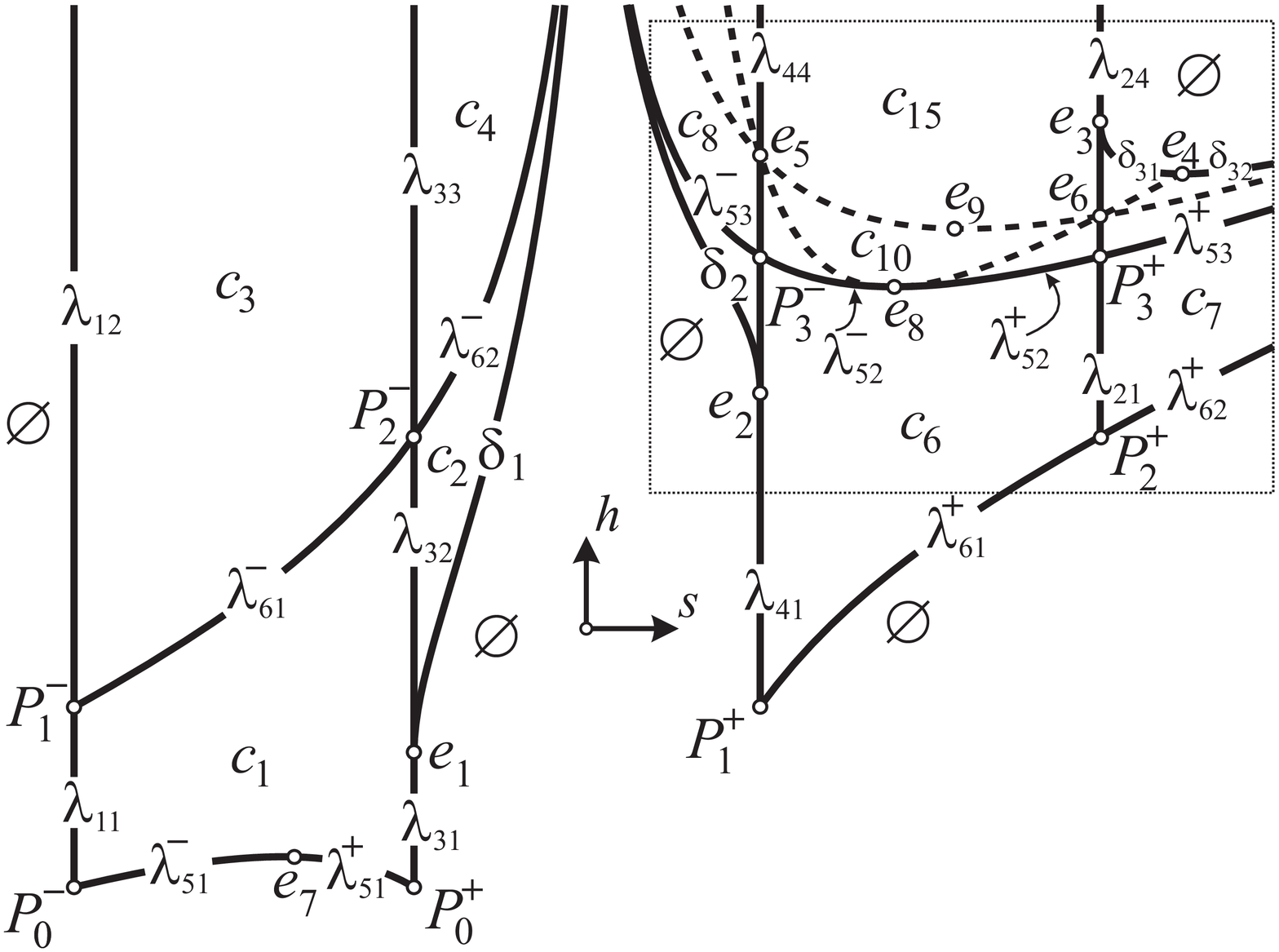}\  \includegraphics[width=0.41\textwidth,keepaspectratio]{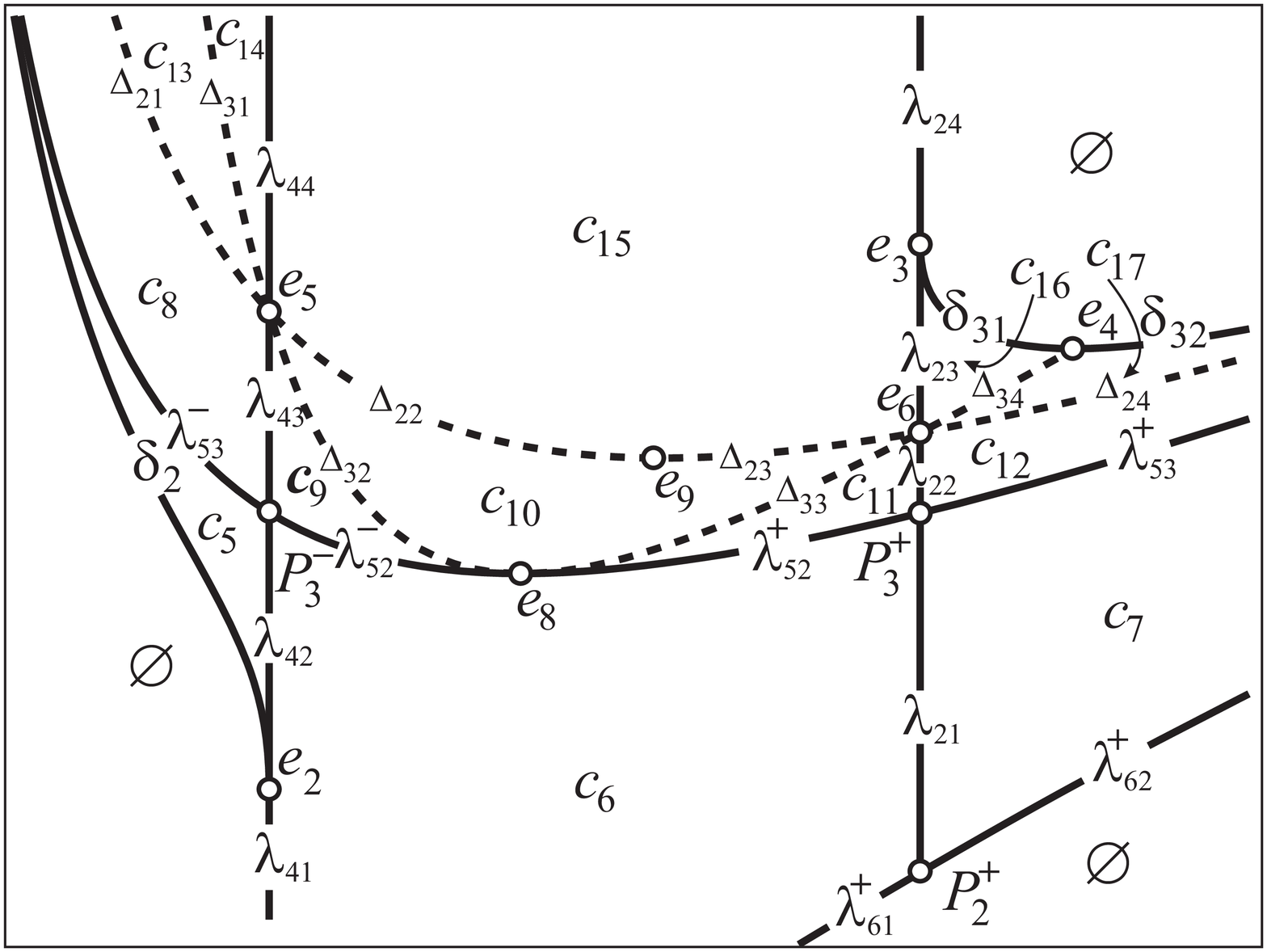}
\caption{The diagram of the subsystem $\mo$ on the $(s,h)$-plane and its fragment.}\label{fig_bifset3}
\end{figure}

\begin{figure}[!htbp]
\centering
\includegraphics[width=0.6\textwidth,keepaspectratio]{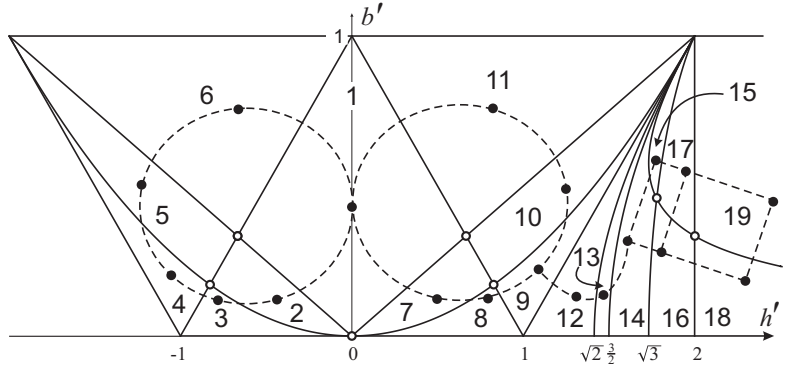}
\caption{The separation set and the paths of the look-up.}\label{fig_chart}
\end{figure}

\begin{figure}[!htbp]
\centering
\includegraphics[width=\cofb\textwidth,keepaspectratio]{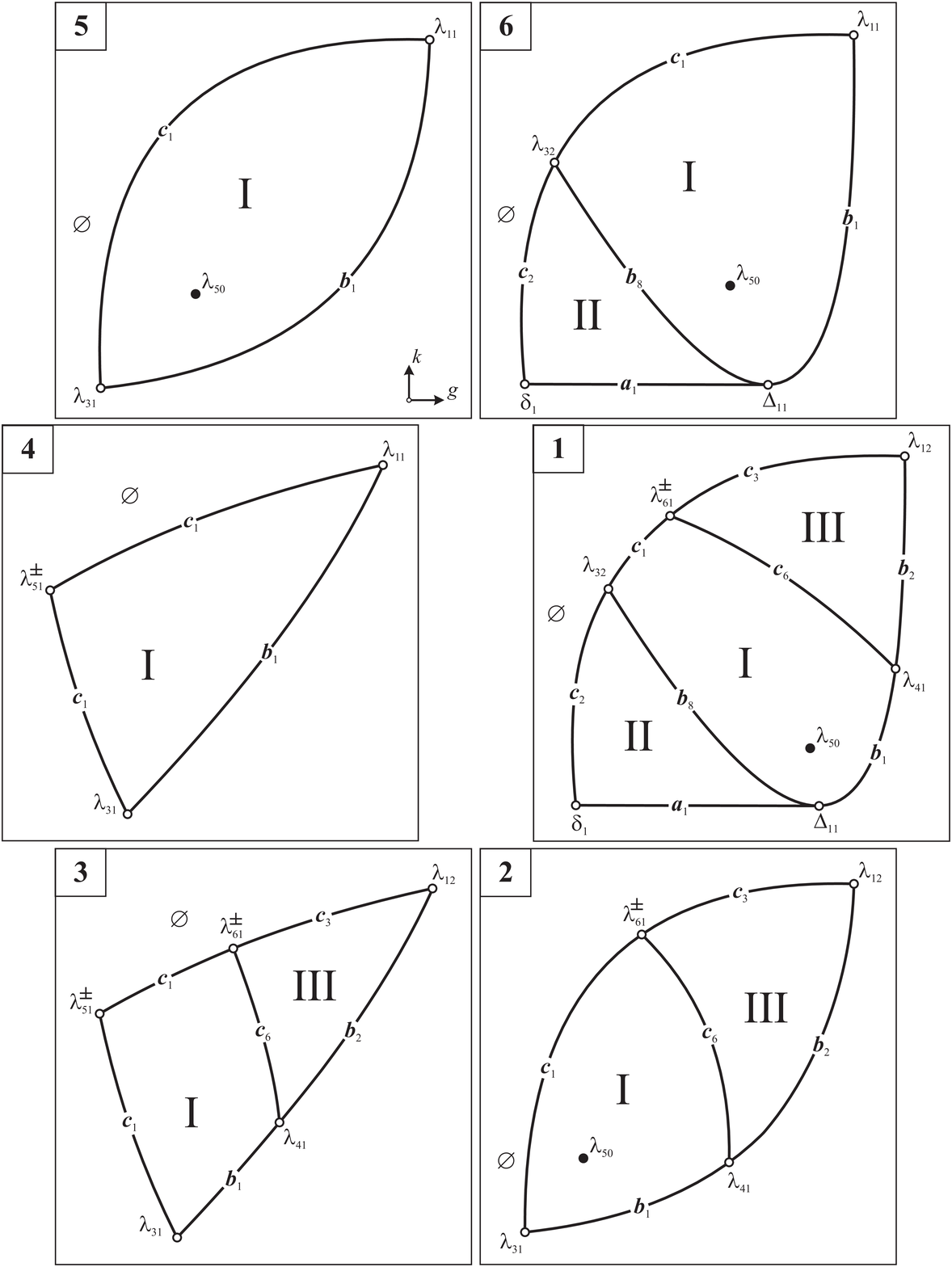}
\caption{Equipped diagrams: ``left circle''.}\label{fig_leftc1}
\end{figure}

\begin{figure}[!htbp]
\centering
\includegraphics[width=\cofb\textwidth,keepaspectratio]{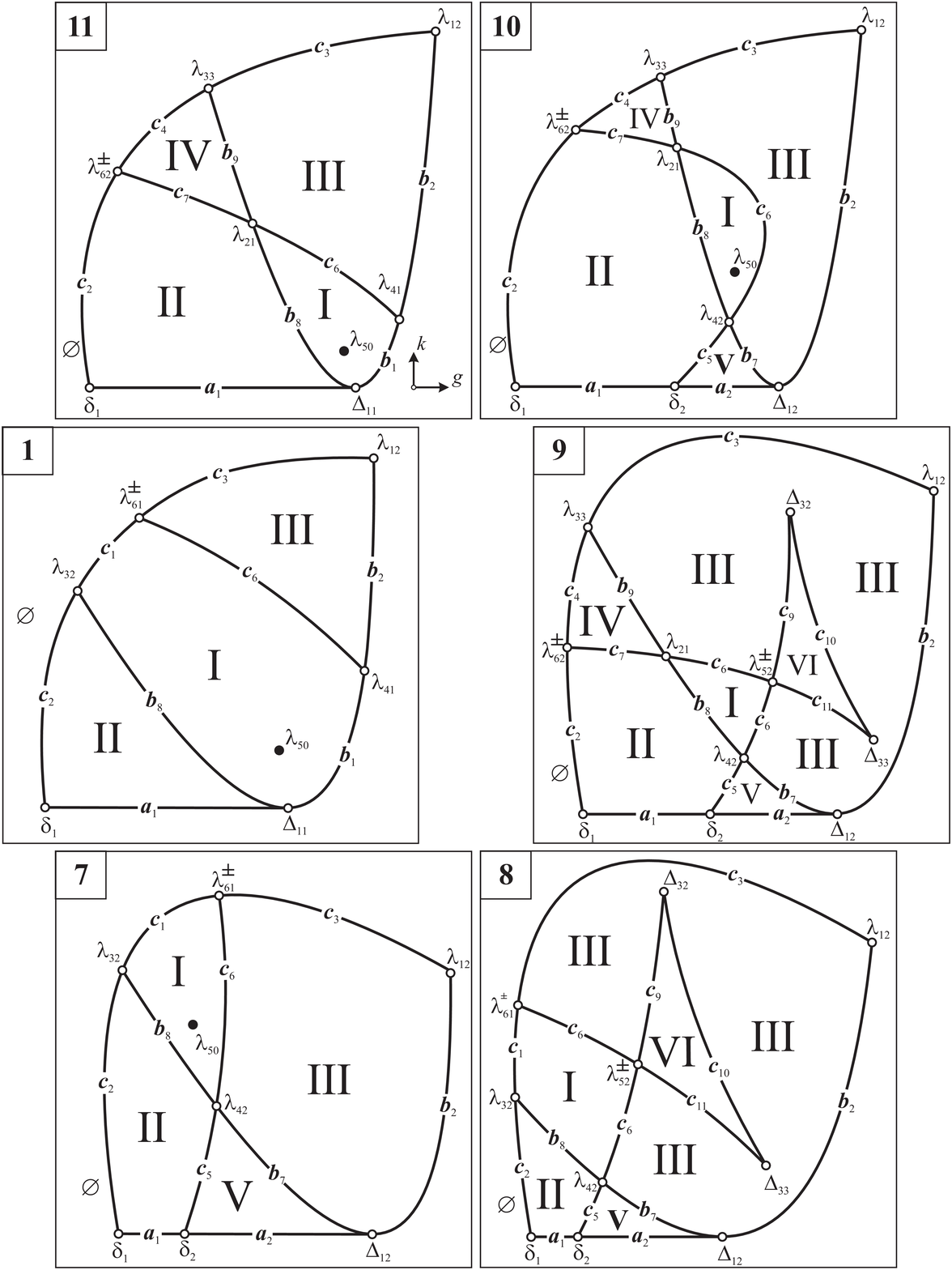}
\caption{Equipped diagrams: ``right circle''.}\label{fig_rightc1}
\end{figure}

\begin{figure}[!htbp]
\centering
\includegraphics[width=\cofs\textwidth,keepaspectratio]{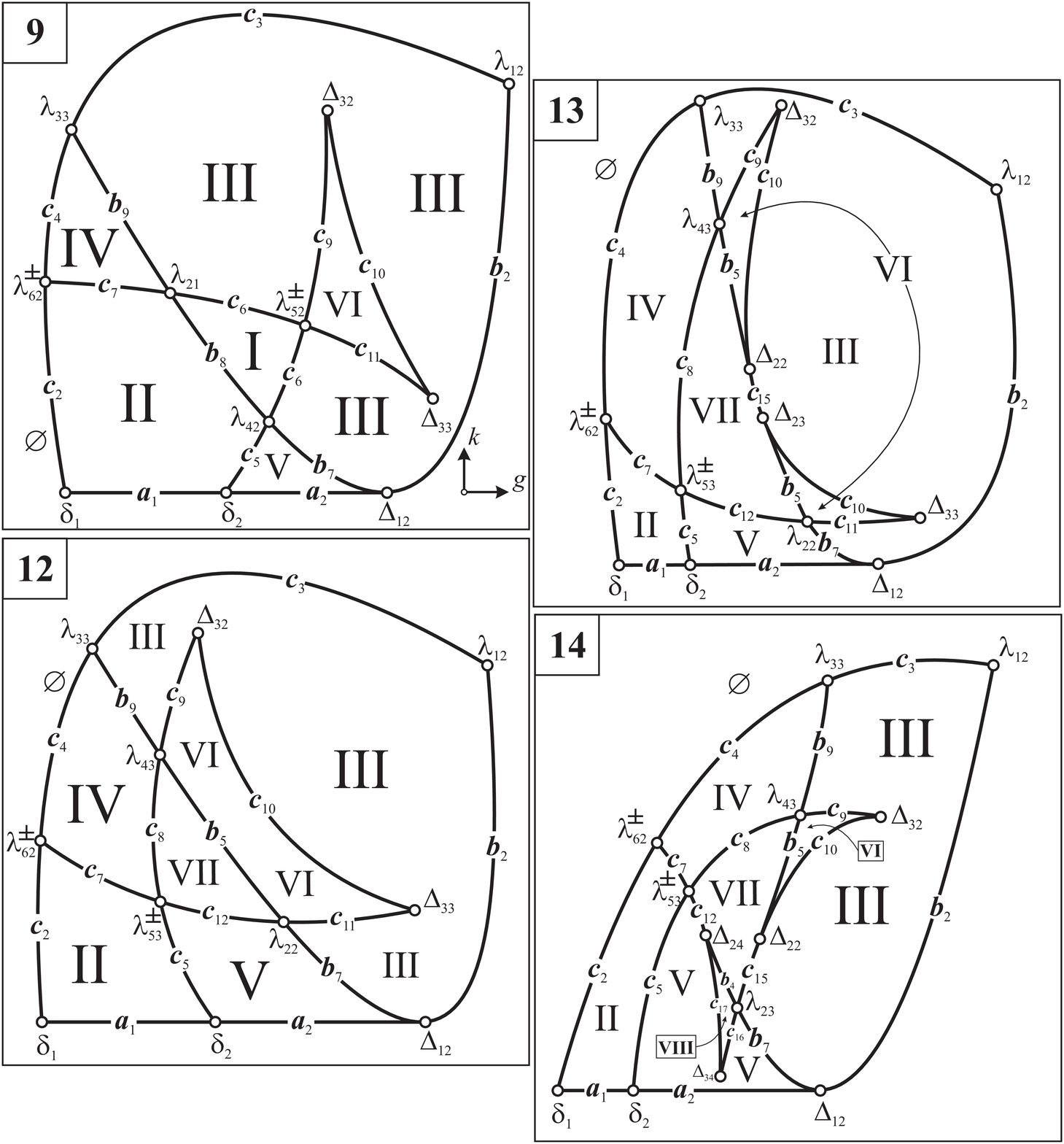}
\caption{Equipped diagrams: ``line''.}\label{fig_line1}
\end{figure}

\begin{figure}[!htbp]
\centering
\includegraphics[width=\cofs\textwidth,keepaspectratio]{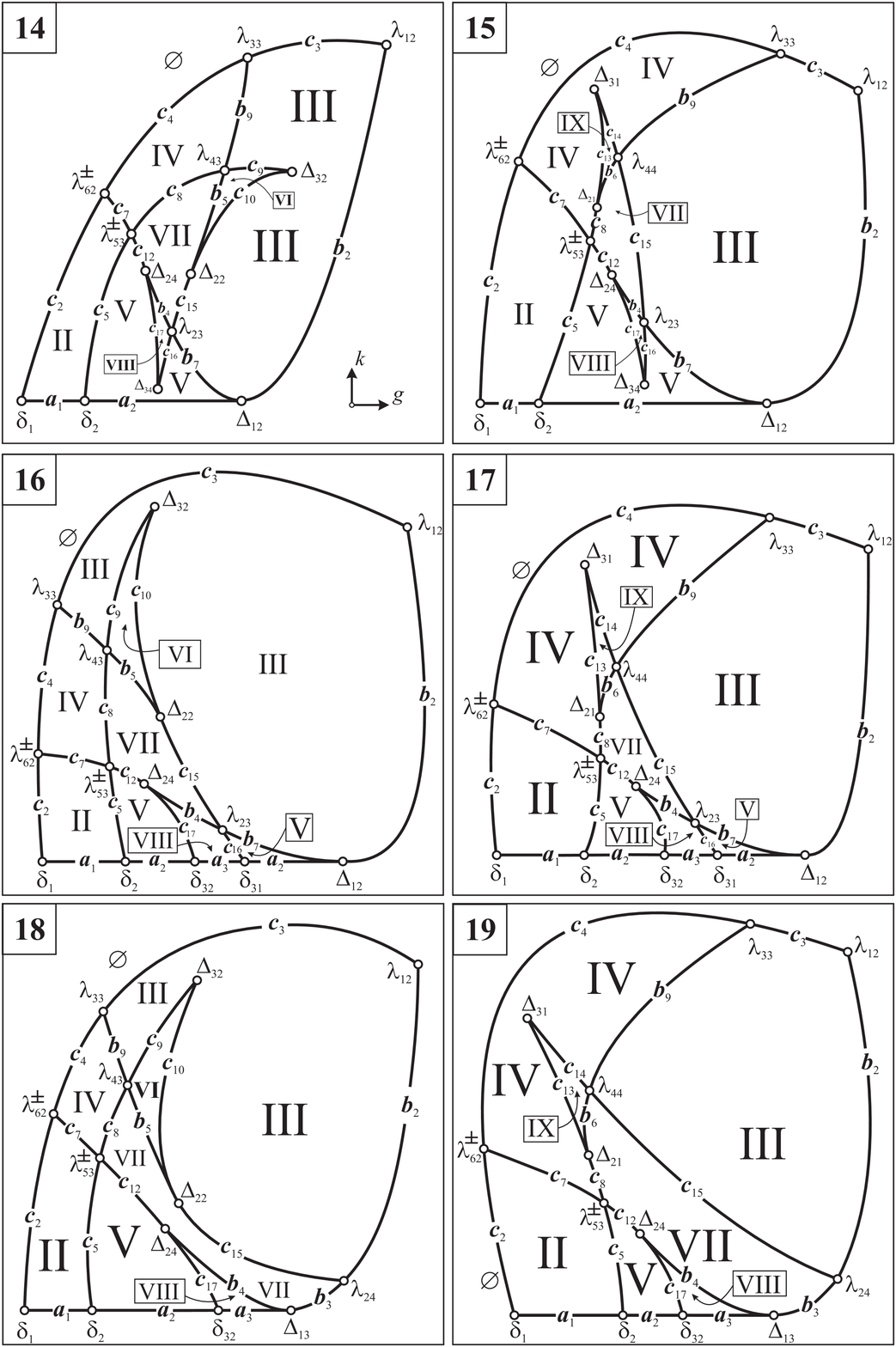}
\caption{Equipped diagrams: ``block''.}\label{fig_block1}
\end{figure}

\begin{figure}[htp]
\centering
\includegraphics[width=\cofb\textwidth,keepaspectratio]{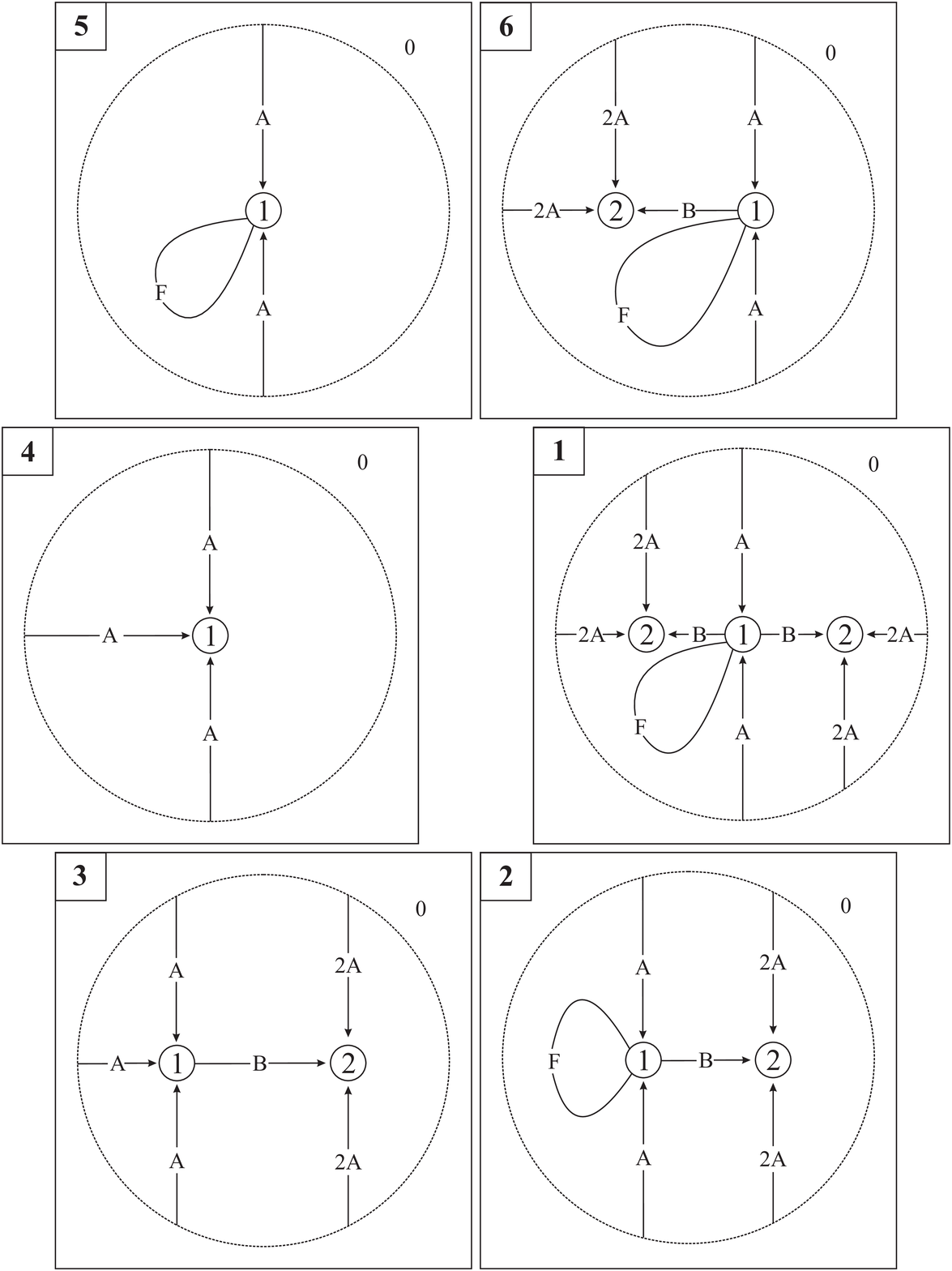}
\caption{Net invariants: ``left circle''.}\label{fig_leftcnet}
\end{figure}

\begin{figure}[htp]
\centering
\includegraphics[width=\cofb\textwidth,keepaspectratio]{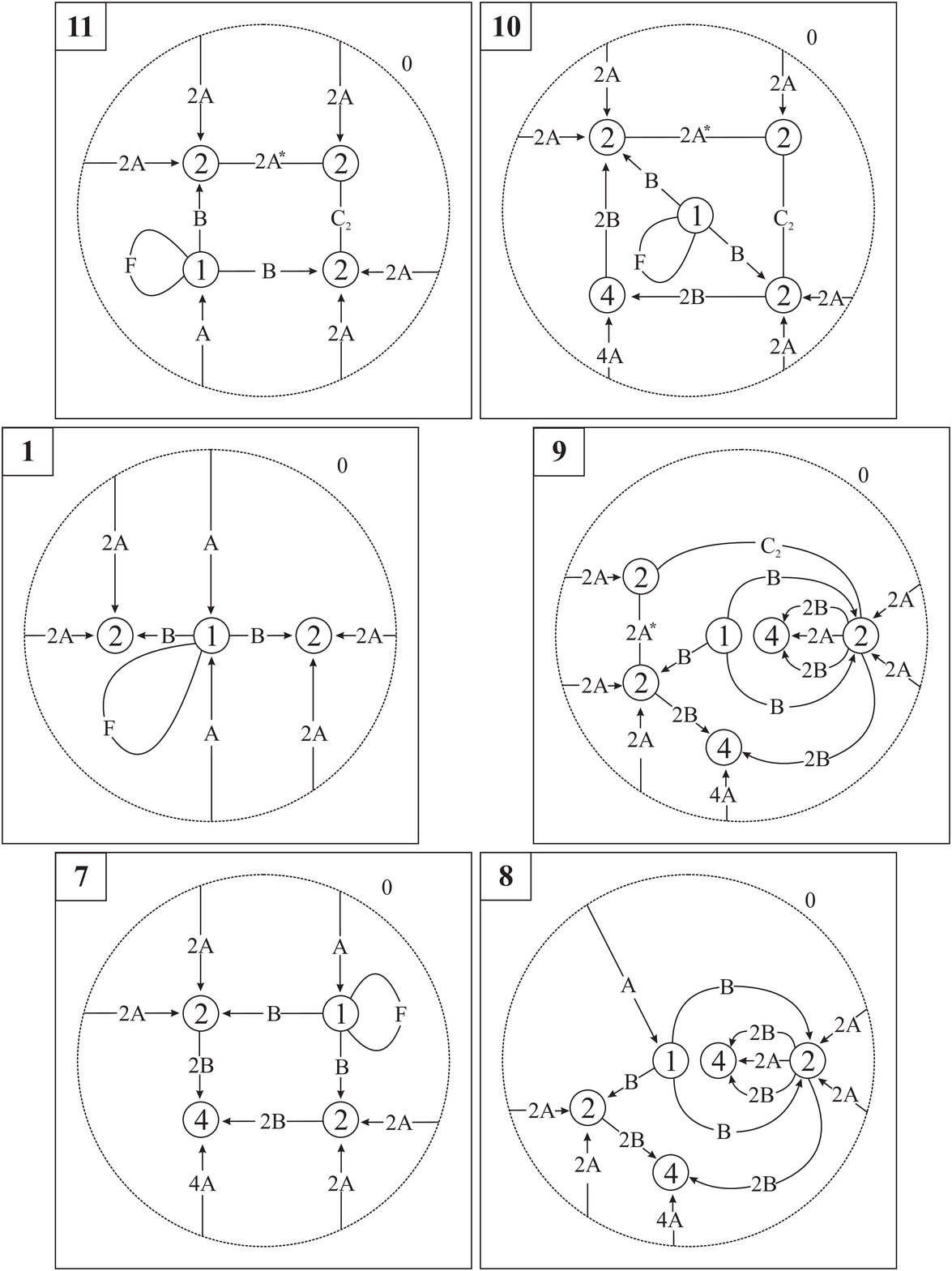}
\caption{Net invariants: ``right circle''.}\label{fig_rightcnet}
\end{figure}

\begin{figure}[htp]
\centering
\includegraphics[width=\cofs\textwidth,keepaspectratio]{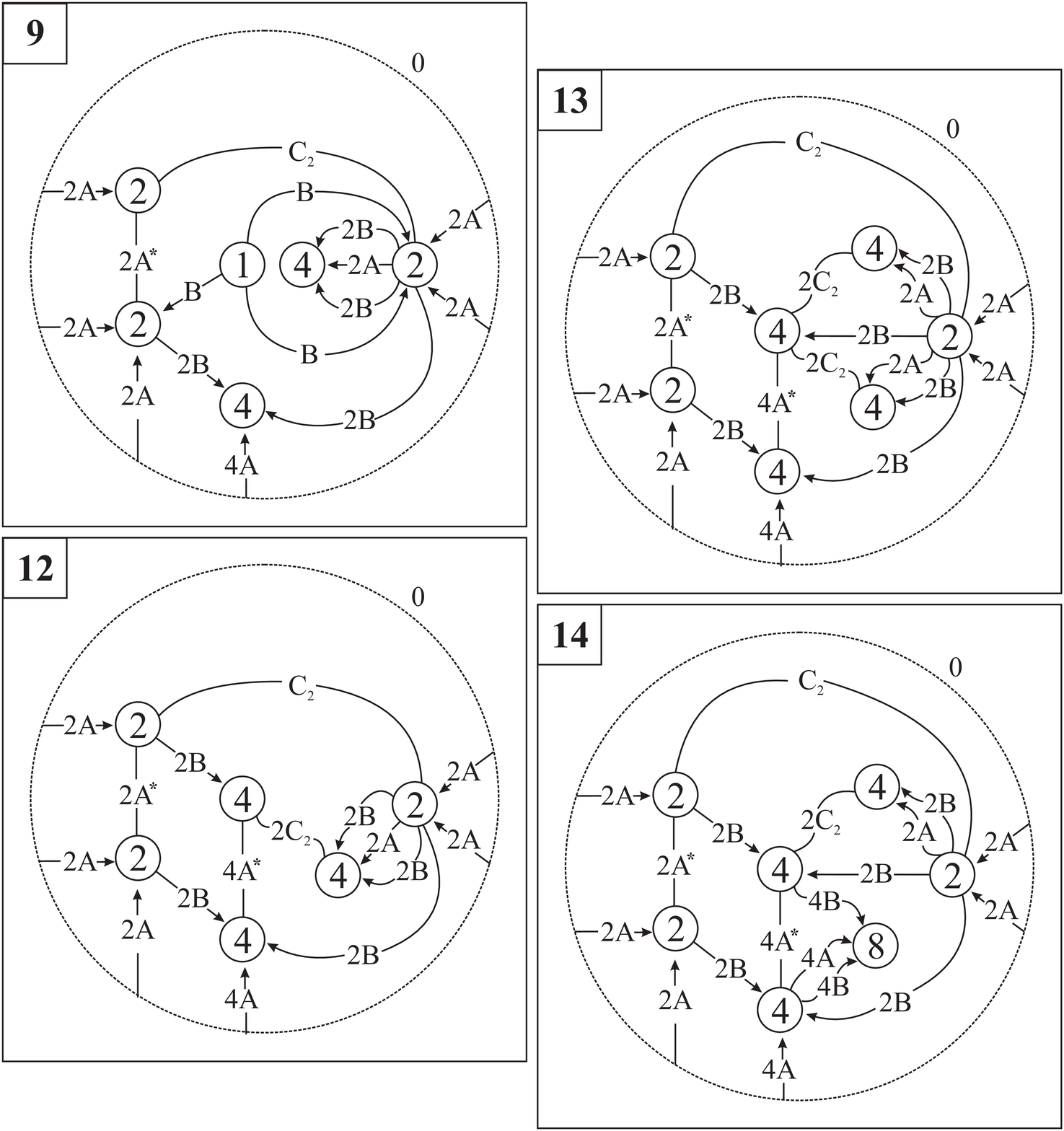}
\caption{Net invariants: ``line''.}\label{fig_linenet}
\end{figure}

\begin{figure}[htp]
\centering
\includegraphics[width=\cofs\textwidth,keepaspectratio]{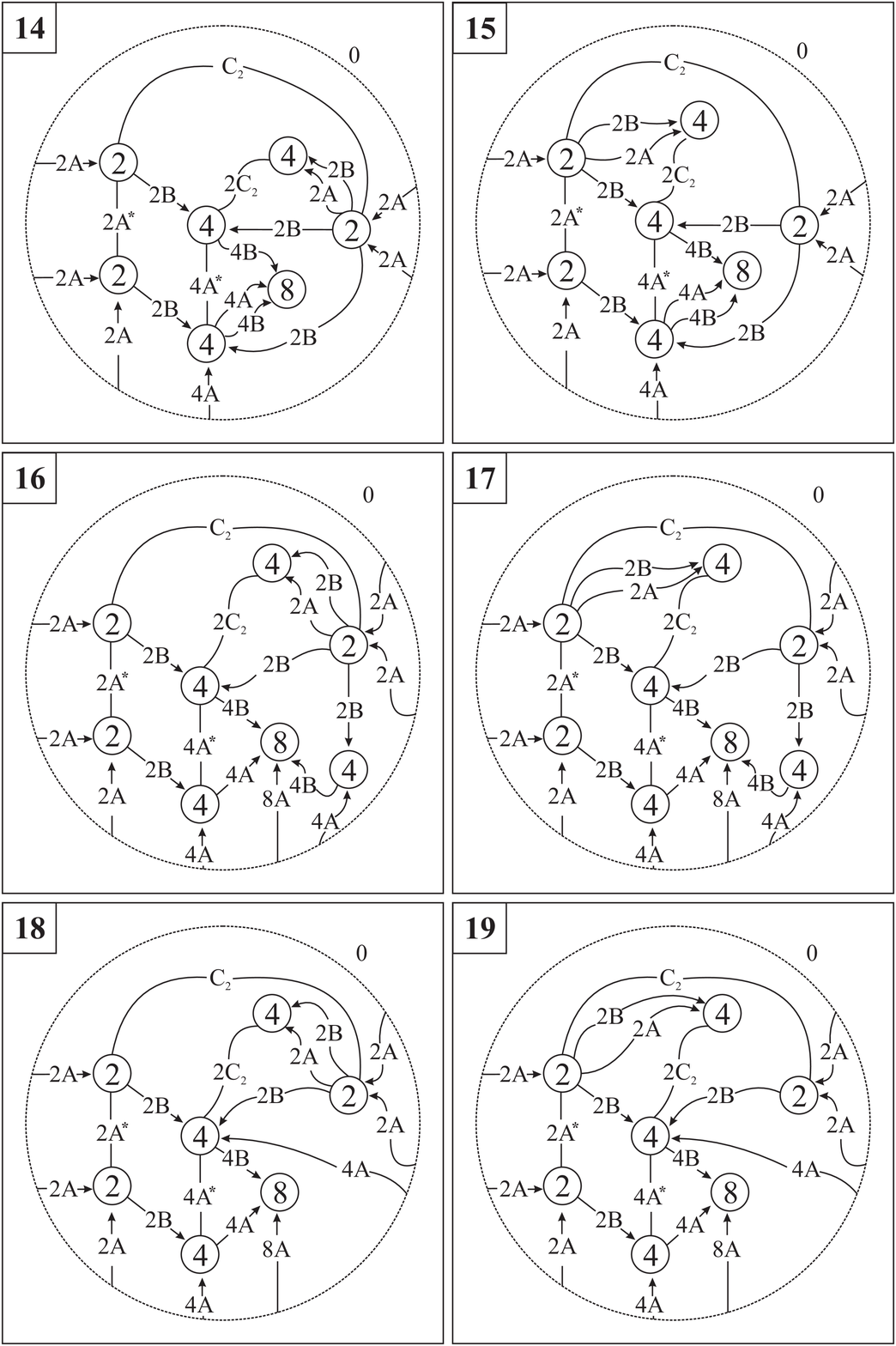}
\caption{Net invariants: ``block''.}\label{fig_blocknet}
\end{figure}

\begin{figure}[ht]
\centering
\includegraphics[width=120mm,keepaspectratio]{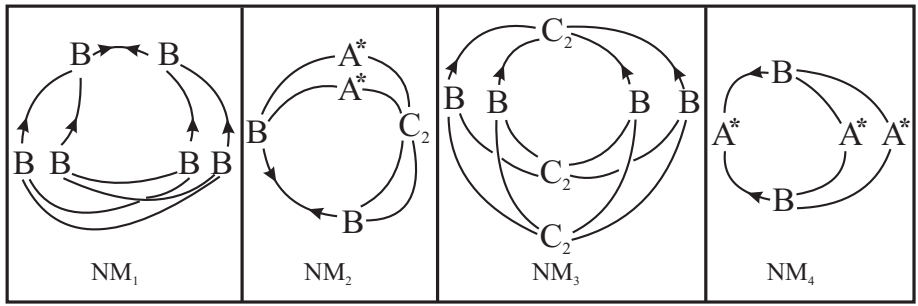}
\caption{Loop molecules of the saddle type singularities of rank 1.}\label{fig_circle}
\end{figure}

\end{document}